\documentclass[11pt]{amsart}
\usepackage[T1]{fontenc}
\usepackage[latin9]{inputenc}
\usepackage{amsthm}
\usepackage{amssymb}
\usepackage{graphicx}
\makeatletter
\theoremstyle{plain}
\newtheorem{thm}{\protect\theoremname}[section]
  \theoremstyle{plain}
  \newtheorem{prop}[thm]{\protect\propositionname}
  \theoremstyle{plain}
  \newtheorem{cor}[thm]{\protect\corollaryname}
  \theoremstyle{plain}
  \newtheorem{lem}[thm]{\protect\lemmaname}
  \theoremstyle{definition}
  \newtheorem{defn}[thm]{\protect\definitionname}
 \theoremstyle{plain}

\makeatother

  \providecommand{\corollaryname}{Corollary}
  \providecommand{\definitionname}{Definition}
  \providecommand{\lemmaname}{Lemma}
  \providecommand{\propositionname}{Proposition}
\providecommand{\theoremname}{Theorem}
\providecommand{\conjecturename}{Conjecture}

\begin{document}
\def\COMMENT#1{}

\global\long\def\labelenumi{(\roman{enumi})}

\def\noproof{{\unskip\nobreak\hfill\penalty50\hskip2em\hbox{}\nobreak\hfill%
        $\square$\parfillskip=0pt\finalhyphendemerits=0\par}\goodbreak}
\def\endproof{\noproof\bigskip}
\newdimen\margin   
\def\textno#1&#2\par{%
    \margin=\hsize
    \advance\margin by -4\parindent
           \setbox1=\hbox{\sl#1}%
    \ifdim\wd1 < \margin
       $$\box1\eqno#2$$%
    \else
       \bigbreak
       \hbox to \hsize{\indent$\vcenter{\advance\hsize by -3\parindent
       \sl\noindent#1}\hfil#2$}%
       \bigbreak
    \fi}
\def\proof{\removelastskip\penalty55\medskip\noindent{\bf Proof. }}

\newcommand{\Int}{{\rm Int}}
\newcommand{\seq}{{\rm seq}}
\newcommand{\Prob}{\mathbb{P}}
\newcommand{\E}{\mathbb{E}}

\title[A domination algorithm for $\{0,1\}$-instances of TSP]{A domination algorithm for $\{0,1\}$-instances of the travelling salesman problem}

\author{Daniela K\"uhn, Deryk Osthus and Viresh Patel}
\begin{abstract}
We present an approximation algorithm for $\{0,1\}$-instances of the travelling salesman problem which performs well with respect to combinatorial dominance. More precisely,
we give a polynomial-time algorithm which has domination ratio $1-n^{-1/29}$. In other words, given a $\{0,1\}$-edge-weighting of the complete graph $K_n$ on $n$ vertices, our 
algorithm outputs a Hamilton cycle $H^*$ of $K_n$ with the following property:
the proportion of Hamilton cycles of $K_n$ whose weight is smaller than that of $H^*$ is at most $n^{-1/29}$.
Our analysis is based on a martingale approach.
Previously, the best result in this direction was a polynomial-time algorithm with domination ratio $1/2-o(1)$ for arbitrary edge-weights.
We also prove a hardness result showing that, if the Exponential Time Hypothesis holds, there exists a constant $C$ such that $n^{-1/29}$  cannot be replaced by $\exp(-(\log n)^C)$ in the result above.
\end{abstract}


\thanks{D.~K\"uhn, D.~Osthus and V.~Patel were supported by the EPSRC, grant no.~EP/J008087/1. 
D.~K\"uhn was also supported by the ERC, grant no.~258345.
D.~Osthus was also supported by the ERC, grant no.~306349.} 

\maketitle

\section{Introduction\label{sec:intro}}

Many important combinatorial optimization problems are known to be NP-hard, and this has led to a vast body of research in approximation algorithms. 
One well-known way to measure the performance of an approximation algorithm is to consider its approximation ratio, i.e.~the cost ratio  of the approximate solution to an optimal solution in the worst case. Another is to consider the proportion of all feasible solutions that are worse than the approximate solution in the worst case. The two measures should be viewed as complementary as there are examples of approximation algorithms that perform well with respect to one measure but badly with respect to the other. It is the latter measure, called \emph{combinatorial dominance}, that we consider in this paper.

In general, the \emph{domination ratio} of an approximation algorithm $A$ for an optimization problem $P$  is the largest $r =r(n)$ such that for each instance $I$ of $P$ of size $n$, $A$ outputs a solution that is not worse than an $r$ proportion of all feasible solutions. The study of this approximation measure was initiated by Glover and Punnen in~\cite{GlovPun}, where they analysed the domination ratio of various heuristics for the travelling salesman problem.

\subsection{Travelling salesman problem}

Let us begin by formally defining the travelling salesman problem.
Let $K_n=(V_n, E_n)$ be the complete graph on $n$ vertices and let $\mathcal{H}_n$ be the set of all $(n-1)!/2$ Hamilton cycles of $K_n$.
For an edge weighting $w$ of $K_n$ (i.e.~a function $w:E_n \rightarrow \mathbb{R}$) and a subgraph $G=(V,E)$ of $K_n$, we define 
\[
w(G):= \sum_{e \in E}w(e).
\]
The \emph{travelling salesman problem (TSP)} is the following algorithmic problem: given an instance $(n,w)$ of TSP, where $n$ is a positive integer and $w$ is an edge weighting of $K_n$, determine a Hamilton cycle $H^*$ of $K_n$ which satisfies
\[
w(H^*) = \min_{H \in \mathcal{H}_n}w(H).
\]
The \emph{asymmetric travelling salesman problem (ATSP)} is a directed version of TSP, in which one considers Hamilton cycles in complete directed graphs and in which the weights
of the two oppositely directed edges between two vertices are allowed to be different from each other.

\subsection{TSP and approximation ratio}
We now give some brief background on TSP and approximation ratio.
It is well known that TSP is NP-hard \cite{GarJon}, and indeed, NP-hard to approximate to within a constant factor \cite{SahGonz}. On the other hand, Christofides \cite{Christ} gave a $3/2$-approximation algorithm for metric-TSP, that is, TSP in which the edge-weights of $K_n$ satisfy the triangle inequality. However, even for $\{1,2\}$-TSP, Papadimitriou and Vempala \cite{PapVem} showed there is no $\frac{220}{219}$-approximation algorithm unless
$P=NP$.%
\COMMENT{There are better bounds than 220/219 for metric TSP I think, but the one here is for $\{1,2\}$-TSP}
Here $\{1,2\}$-TSP is the special case of TSP where all edge-weights are either $1$ or $2$, and this is in fact a special case of metric TSP. 
Arora \cite{Arora} and Mitchell \cite{Mit} independently gave a PTAS for Euclidean-TSP, a special case of metric-TSP in which the edge-weights of $K_n$ arise as the distances between vertices that have been embedded in Euclidean space of fixed dimension. 

\subsection{TSP and combinatorial dominance}

A TSP algorithm is an algorithm which, given any instance $(n,w)$ of TSP, outputs some Hamilton cycle of $K_n$. For $r \in [0,1]$ and $(n,w)$ a fixed instance of TSP, we say
that a TSP algorithm $A$  has \emph{domination ratio at least $r$ for $(n,w)$} if,
given $(n,w)$ as input, the algorithm $A$ outputs a Hamilton cycle $H^*$ of $K_n$ satisfying 
\[
\frac{|\{ H \in \mathcal{H}_n: w(H^*) \leq w(H) \}|}{|\mathcal{H}_n|} \geq r.
\]
The \emph{domination ratio} of a TSP algorithm is the maximum $r$ such that the algorithm has domination ratio at least $r$ for all instances $(n,w)$ of TSP. (Thus, the aim is to have a domination ratio close to one.) We often refer to a TSP algorithm as a TSP-domination algorithm to indicate our intention to evaluate its performance in terms of the domination ratio.

The notion of combinatorial dominance (although slightly different to above) was first introduced by Glover and Punnen~\cite{GlovPun}. They gave various polynomial-time TSP-domination algorithms and showed
that their algorithms have domination ratio $\Omega(c^n/n!)$ for some constant $c>1$.
Glover and Punnen \cite{GlovPun} believed that, unless $P=NP$, there cannot be a polynomial-time TSP-domination algorithm with domination ratio at least $1/p(n)$ for any polynomial $p$.
This was disproved by Gutin and Yeo \cite{GutYeo2002}, who gave a polynomial-time ATSP-domination algorithm with domination ratio at least $1/(n-1)$.
It was later found that this had already been established in the early 1970's \cite{Rub,Sarv}.
Alon, Gutin, and Krivelevich \cite{AlonGutKriv} asked whether one can achieve a domination ratio $\Omega(1)$.
K{\"u}hn and Osthus~\cite{KuhOst} resolved this question by proving an algorithmic Hamilton decomposition result which,
together with a result by Gutin and Yeo~\cite{GutYeo2001}, gives a polynomial-time ATSP-domination algorithm with domination
ratio at least $1/2-o(1)$.

It is worth noting (see \cite{GutYeoZver}) that some simple, well-known TSP heuristics give the worst possible domination ratio, that is, for certain instances of TSP, these heuristics produce the unique worst (i.e.\ most expensive) Hamilton cycle. In particular this is the case for the greedy algorithm, in which one recursively adds the cheapest edge maintaining a disjoint union of paths (until the end), and the nearest neighbour algorithm, in which one builds (and finally closes) a path by recursively adding the cheapest neighbour of the current end-vertex.\COMMENT{The double tree heuristic also has the worst possible domination ratio.}  Other algorithms have been shown to have small domination ratio \cite{PunMarKab}, e.g. Christofides' algorithm \cite{Christ} has domination ratio at most $\lfloor n/2 \rfloor!/\frac{1}{2}(n-1)! = \exp(-\Omega(n))$, even for instances of TSP satisfying the triangle inequality.


Our main result gives an algorithm with domination ratio of $1 - o(1)$ for the TSP problem restricted to $\{0,1\}$-instances, i.e.~for instances where all the edge-weights lie in $\{0,1\}$. 
Note that this clearly implies the same
result whenever the weights may take two possible values; in particular it implies the same result for $\{1,2\}$-instances. 
The latter is the more usual formulation, but for our algorithm we find it more natural to work with weights in  $\{0,1\}$
rather than in  $\{1,2\}$.
\begin{thm}
\label{thm:main}
There exists an $O(n^5)$-time TSP-domination algorithm which has domination ratio at least $1 - 6n^{-1/28}$ for every $\{0,1\}$-instance $(n,w)$ of TSP.
\end{thm}
In fact we have a TSP-domination algorithm which, for most $\{0,1\}$-instances of TSP, has domination ratio that is exponentially close to $1$.
\begin{thm}
\label{thm:unused}
Fix\COMMENT{The statement of this theorem has been reworded as compared to the original version in Section~\ref{sec:alg}} $\eta \in (0,1/2)$. 
There exists an $O(n^5)$-time TSP-domination algorithm which has domination ratio at least $1 - O(\exp(-\eta^4n/10^4))$ for every 
$\{0,1\}$-instance $(n,w)$ of TSP satisfying $w(K_n)= d\binom{n}{2}$ for some $d \in [\eta,1- \eta]$.
\end{thm}

We use a combination of algorithmic and probabilistic techniques as well as some ideas from extremal combinatorics to establish the results above.
On the hardness side, assuming the Exponential Time Hypothesis (discussed in Section~\ref{sec:hard}), we show that there is no TSP-domination algorithm for $\{0,1\}$-instances that substantially improves the domination ratio in Theorem~\ref{thm:main} and, in particular, there is no TSP-domination algorithm that achieves the domination ratio of Theorem~\ref{thm:unused} for general $\{0,1\}$-instances of TSP. 
We can prove weaker bounds assuming $P \not= NP$. 
\begin{thm} \label{thm:hardness1} 
$ $
\begin{enumerate}
\item[(a)] Fix $\varepsilon \in (0, \frac{1}{2})$. If $P \not= NP$, then there is no polynomial-time TSP-domination algorithm which, for all $\{0,1\}$-instances $(n,w)$ of TSP, has domination ratio at least $1-\exp(-n^{\varepsilon})$.
\item[(b)] There exists a constant $C>0$ such that if the Exponential Time Hypothesis holds, then there is no polynomial-time TSP-domination algorithm which, for all $\{0,1\}$-instances $(n,w)$ of TSP, has domination ratio 
at least $1-\exp(-(\log n)^C)$.\COMMENT{TSP13: The statment of (b) has changed a little, and corresponding changes have been made in the proof}
\end{enumerate}
\end{thm}
Previously, Gutin, Koller and Yeo \cite{GutKolYeo} proved a version of Theorem~\ref{thm:hardness1}(a) for general instances of TSP, and we use some of the ideas from their reduction in our proof.

\subsection{Combinatorial dominance and other algorithmic problems}
Alon, Gutin and Krivelevich \cite{AlonGutKriv} gave algorithms with large domination ratios for several optimization problems. For example, they gave a $(1-o(1))$-domination algorithm for the minimum partition problem and an $\Omega(1)$-domination algorithm for the max-cut problem. Berend, Skiena and Twitto \cite{Ber} introduced a notion of combinatorial dominance for constrained problems, where not all permutations, subsets or partitions form feasible solutions. They analysed various algorithms for the maximum clique problem, the minimum subset cover and the maximum subset sum problem with respect to their notion of combinatorial dominance. 
See e.g.~\cite{GutVainYeo,Kol} for further results on domination analysis applied to other optimization problems.

\subsection{Organisation of the paper}
In the next section, we introduce some basic terminology and notation that we use throughout the paper. In Section~\ref{sec:basicalg}, we present an algorithm, which we call Algorithm~A, that will turn out to have a large domination ratio for `most' $\{0,1\}$-instances of TSP. We define it as a randomized algorithm and describe how it can be derandomized using the method of conditional expectations of Erd\H{o}s and Selfridge.
In Section~\ref{sec:mar}, we develop the probabilistic tools needed to evaluate the domination ratio of Algorithm~A. In particular, our approach is
based on a martingale associated with a random TSP tour.
We begin Section~\ref{sec:alg} by evaluating the domination ratio of Algorithm~A and proving Theorem~\ref{thm:unused}. We also present and evaluate two other algorithms, Algorithm~B and Algorithm~C, which have large domination ratios for $\{0,1\}$-instances of TSP where Algorithm~A does not work well (roughly speaking, this is the case when almost all weights are $0$ or almost all
weights are $1$). The analysis of Algorithms~B and~C is based on a result from extremal combinatorics. We conclude Section~\ref{sec:alg} by proving Theorem~\ref{thm:main}. Section~\ref{sec:hard} is devoted to the proof of Theorem~\ref{thm:hardness1}. We end with concluding remarks and an open problem in Section~\ref{sec:conc}.

\section{Preliminaries\label{sec:preliminaries}}

We use standard graph theory notation.
Let $G=(V,E)$ be a graph. 
We sometimes write $V(G)$ and $E(G)$ for the vertex and edge set of $G$ and we write $e(G)$ for number of edges of $G$.

We write $H \subseteq G$ to mean that $H$ is a subgraph of $G$.
For $U \subseteq V$, we write $G[U]$ for the graph induced by $G$ on $U$, $E_G(U)$ for the set of edges in $G[U]$, and $e_G(U)$ for the number of edges in $G[U]$. Similarly for  $A,B \subseteq V$ not necessarily disjoint, we write $G[A,B]$ for the graph given by $G[A,B] := (A \cup B, E_G(A,B))$, where
\[
E_G(A,B) := \{ab \in E: a \in A, b \in B \}.
\]
Define $e_G(A,B):=|E_G(A,B)|$. For $v \in V$, we write $\deg_G(v)$ for the number of neighbours of $v$ in $G$. A set $D \subseteq V$ is called a \emph{vertex cover} if $V \setminus D$ is an independent set in $G$, i.e.\ if $E_G(V \setminus D) = \emptyset$. 

For any set $A$, we write $A^{(2)}$ for the set of all unordered pairs in $A$.  Thus if $A$ is a set of vertices, then $A^{(2)}$ is the set of all edges between the vertices of $A$.

Recall that we write $K_n = (V_n, E_n)$ for the complete graph on $n$ vertices, where $V_n$ and $E_n$ is the vertex and edge set respectively of $K_n$. 
A function $w: E_n \rightarrow \mathbb{R}$ is called a \emph{weighting} of $K_n$. For $E \subseteq E_n$, and more generally, for a subgraph $G=(V,E)$ of $K_n$, we define
\[
w(G) = w(E) := \sum_{e \in E}w(e) \:\:\:\text{and}\:\:\:
w[G] = w[E] := \sum_{e \in E}|w(e)|.
\]

Recall that we write $\mathcal{H}_n$ for the set of all $(n-1)!/2$ Hamilton cycles of $K_n$. We write $\tilde{\mathcal{H}}_n$ for the set of all $(n-1)!$ directed Hamilton cycles of $K_n$. 
For $n$ even, we define $\mathcal{M}_n$ to be the set of all perfect matchings of $K_n$. An \emph{optimal matching} of $K_n$ is any set of $\lfloor n/2 \rfloor$ independent edges in $K_n$; thus an optimal matching is a perfect matching if $n$ is even and is a matching spanning all but one vertex of $K_n$ if $n$ is odd. At certain points throughout the course of the paper, we shall distinguish between the cases when $n$ is odd and even. The case when $n$ is odd requires a little extra care, but the reader loses very little by focusing on the case when $n$ is even.



We have already defined the domination ratio of a TSP-domination algorithm, but we give here an equivalent reformulation which we shall use henceforth.
Note that, for $r \in [0,1]$ and $(n,w)$ a fixed instance of TSP, a TSP algorithm $A$ has \emph{domination ratio at least $r$ for $(n,w)$} if and only if,
given $(n,w)$ as input, the algorithm $A$ outputs a Hamilton cycle $H^*$ of $K_n$ satisfying 
$$\Prob(w(H) < w(H^*)) \leq 1-r,$$
where $H$ is a uniformly random Hamilton cycle from $\mathcal{H}_n$. 
\COMMENT{TSP14: Changed paragraph above.}


Note that for the TSP problem, if $\lambda>0$, then two instances $(n,w)$ and $(n, \lambda w)$ are completely equivalent, and so by suitably scaling $w$, we can and shall always assume that $w: E_n \rightarrow [-1,1]$. 

While the TSP problem is known to be NP-hard, the corresponding problem for optimal matchings is polynomial-time solvable \cite{Ed1}.

\begin{thm} \label{thm:matalg}
There exists an $O(n^4)$-time\COMMENT{This is not best available, but finding a min weight optimal matching does not dominate the time complexity of our algorithm anyway.} algorithm which, given $(n,w)$ as input (where $n \geq 2$  and $w$ is a weighting of $K_n$), outputs an optimal matching 
 $M^*$ of $K_n$, such that
\[
w(M^*) = \min_{M \in \mathcal{M}_n} w(M).
\] 
\end{thm}

We call this algorithm the \emph{minimum weight optimal matching algorithm}.


\section{An algorithm\label{sec:algorithm}}
\label{sec:basicalg}

We now informally describe a simple polynomial-time algorithm that turns out to give a domination ratio close to one\COMMENT{TSP11: reworded slightly} for many instances of TSP. Given an instance $(n,w)$ of TSP, we apply the minimum weight optimal matching algorithm to find an optimal matching $M^*$ of $K_n$ of minimum weight. We then extend $M^*$ to a Hamilton cycle $H^*$ of $K_n$ using a randomized approach. This last step can be transformed into a deterministic polynomial-time algorithm using the method of conditional expectations of Erd\H{o}s and Selfridge \cite{ErdSelf,ErdSpe,MotRag}.

Note that in this section and the next, we allow our weightings to take values in the interval $[-1,1]$. This does not make the analysis more difficult than in the $\{0,1\}$-case.%
   \COMMENT{TST14: new sentence}

\begin{lem}
\label{lem:derand}
There exists an $O(n^5)$-time algorithm which, given an instance $(n,w)$ of TSP (with $n \geq 3$) and any optimal matching $\hat{M}$ of $K_n$ as input,  outputs a Hamilton cycle $\hat{H}$ of $K_n$ satisfying
\begin{equation} \label{eq:exp}
w(\hat{H}) \leq \left( 1 - \frac{1}{n-2} \right) w(\hat{M}) + \frac{1}{n-2}w(K_n) + \rho(n),
\end{equation}
where $\rho(n) = 1$ if $n$ is odd and $\rho(n)=0$ if $n$ is even.
\end{lem} 
\begin{proof}
Assume first that $n \geq 4$ is even. For $G \subseteq K_n$, we write $H_G$ for a uniformly random Hamilton cycle from $\mathcal{H}_n$ that contains all edges of $G$. We have
\begin{align*}
\E(w( H_{\hat{M}} )) 
&= \sum_{ e \in E_n } \Prob( e \in H_{\hat{M}})w(e) 
= \sum_{ e \in E(\hat{M}) } w(e) + \sum_{e \in E_n \setminus E(\hat{M})} \frac{1}{n-2}w(e) \\
&= w(\hat{M}) + \frac{1}{n-2}(w(K_n) - w(\hat{M})) \\
&= \left(1 - \frac{1}{n-2} \right)w(\hat{M}) + \frac{1}{n-2}w(K_n).
\end{align*}
Let $\hat{H}$ be a Hamilton cycle of $K_n$ such that $w(\hat{H}) \leq \mathbb{E}(w( H_{\hat{M}} ))$; thus $\hat{H}$ satisfies (\ref{eq:exp}).%
\COMMENT{TSP10: last sentence rephrased so that it's not an existence statement}
 It remains for us to show that we can find such a Hamilton cycle in $O(n^5)$-time. The following claim provides a subroutine that we shall iteratively apply to obtain the desired algorithm. 

By a \emph{non-trivial path}, we mean a path with at least two vertices. If $G=(V_n, E) \subseteq K_n$ is the union of vertex-disjoint non-trivial paths, let $J(E)$ denote those edges of $K_n$ that join the end-vertices of two distinct paths of $G$ together.

\medskip

\noindent
\textbf{Claim.} \emph{Let $w$ be a weighting of $K_n$ and let $G=(V_n,E)$ be a union of $i \geq 2$ vertex-disjoint non-trivial paths. 
Then in $O(n^4)$-time, we can find $e^* \in J(E)$ such that
\[
\E(w(H_{G \cup e^*})) \leq \E(w(H_G)).
\]
Furthermore $G \cup e^*$ is the union of $i-1$ vertex-disjoint non-trivial paths.}

\smallskip

\noindent
First note that we can determine $J(E)$ in $O(n^2)$-time and that for each $e \in J(E)$, $G \cup e$ is the disjoint union of $i-1$ non-trivial paths. Furthermore, for each $e \in J(E)$, we can compute $\E(w(H_{G \cup e}))$ in $O(n^2)$-time. Indeed, assuming first that $i>2$, note that $\Prob(e' \in H_{G \cup e}) = \frac{1}{2(i-2)}$
 for all $e' \in J(E \cup \{ e \})$ since $J(E \cup \{ e \})$ is regular of degree $2(i-2)$ and each edge at a given vertex is equally likely to be in $H_{G \cup e}$. Now we have
\begin{align*}
\E(w(H_{G \cup e}))
&=w(G \cup e) + \sum_{e' \in J(E \cup \{ e \})} 
\Prob(e' \in H_{G \cup e}) w(e') \\
&= w(G \cup e) + \frac{1}{2(i-2)}
\sum_{e' \in J(E \cup \{ e \})}w(e'),
\end{align*}
so we see that computing $\E(w(H_{G \cup e}))$ takes $O(n^2)$-time. (For $i=2$, we have $\E(w(H_{G \cup e}))
=w(G \cup e) + w(e')$, where $e'$ is the unique edge that closes the Hamilton path $G \cup e$ into a Hamilton cycle.)\COMMENT{TSP10: sentence in brackets added}
Since
\[
\E(w(H_G)) = \frac{1}{|J(E)|}\sum_{e \in J(E)} 
\E(w(H_{G \cup e})),
\]
there exists some $e^*\in J(E)$ such that $\E(w(H_{G \cup e^*})) \leq \mathbb(w(H_G))$. By computing $\E(w(H_{G \cup e}))$ for each $e \in J(E)$, we can determine $e^*$ in $O(n^4)$-time.\COMMENT{In fact can do it in $O(n^3)$-time if we are allowed to store information about calculations we have made.} This proves the claim.


\medskip

\noindent
We now iteratively apply the subroutine from the claim $n/2 - 1$ times. Thus let $A_0 := \hat{M}$, and let $G_0 := (V_n, A_0)$, and for each $i = 1, \ldots, n/2 -1$, let $G_i := (V_n, A_i)$ be obtained from $G_{i-1} := (V_n, A_{i-1})$ by setting $A_i := A_{i-1} \cup \{ e_i \}$, where $e_i$ is obtained by applying the subroutine of the claim to $G_{i-1}$.

By induction, it is clear that $G_i$ is the disjoint union of $n/2 - i \geq 2$ non-trivial paths for $i=0, \ldots, n/2-2$, and so the claim can be applied at each stage. By induction it is also clear for all $i = 1, \ldots, n/2 -1$, that 
$\E(H_{G_i}) \leq \E(H_{\hat{M}})$. Let $\hat{H}$ be the Hamilton cycle obtained by closing the Hamilton path $G_{n/2-1}$. Then we have 
\[
w(\hat{H}) = \E(w(H_{G_{n/2-1}})) 
\leq \E(H_{\hat{M}}),
\]
as required.

The running time of the algorithm is dominated by the $O(n)$ applications of the subroutine from the claim each taking $O(n^4)$-time, giving a running time of $O(n^5)$.

Now consider the case when $n \geq 3$ is odd. Let $(n,w)$ be an instance of TSP with $n \geq 3$ odd and $\hat{M}$ an optimal matching of $K_n = (V_n,E_n)$. We introduce a weighting $w'$ of $K_{n+1} \supseteq K_n$ defined as follows. Let $v$ be the unmatched vertex of $K_n$ in $\hat{M}$ and let $v'$ be the unique vertex in $K_{n+1}$ but not in $K_n$.  Set $w'(e):=w(e)$ for all $e \in E_n$, set $w(v'x):=w(vx)$ for all $x \in V_n \setminus \{ v \}$, and set $w(vv'):=0$. Let $\hat{M}'$ be the perfect matching of $K_{n+1}$ given by $\hat{M}' := \hat{M} \cup \{ vv' \}$. We apply the algorithm for $n$ even to $(n+1,w')$ and $\hat{M}'$ to produce a Hamilton cycle $\hat{H}'$ satisfying
\begin{align*}
w'(\hat{H}') 
&\leq \left( 1 - \frac{1}{n-1} \right) w'(\hat{M}') + \frac{1}{n-1}w'(K_{n+1}) \\
&\leq \left( 1 - \frac{1}{n-1} \right) w(\hat{M}) + \frac{1}{n-1}w(K_{n}) + 1 \\
&\leq \left( 1 - \frac{1}{n-2} \right) w(\hat{M}) + \frac{1}{n-2}w(K_{n}) + 1, 
\end{align*}
where the second inequality follows from the fact that $w(\hat{M}) = w'(\hat{M}')$ and $w'(K_{n+1}) \leq w(K_n) + n-1$. Now contracting the edge $vv'$ in $\hat{H}'$ to give a Hamilton cycle $\hat{H}$ of $K_n$, and noting that $w'(\hat{H}')=w(\hat{H})$, the result immediately follows.
\end{proof}

We end the section by formally describing the steps of our main algorithm, which we call Algorithm~A.

\begin{enumerate}
\item[1.] Using the minimum weight optimal matching algorithm (Theorem~\ref{thm:matalg}), find a minimum weight optimal matching $M^*$ of $K_n$. ($O(n^4)$ time)
\item[2.] Apply the algorithm of Lemma~\ref{lem:derand} to extend $M^*$ to a Hamilton cycle $H^*$ satisfying $w(H^*) \leq (1 - \frac{1}{n-2})w(M^*) + \frac{1}{n-2}w(K_n) + \rho(n)$.%
\end{enumerate}
(Recall that $\rho(n) =1$ if $n$ is odd and $\rho(n)=0$ otherwise.)%
\COMMENT{TSP12: sentence above added}
So Algorithm~A has running time $O(n^5)$.
In the next few sections, we evaluate the performance of this algorithm for certain instances of TSP.


\section{Martingale estimates \label{sec:mar}}

Our next task is to find for each instance $(n,w)$ of TSP and each $r \in [0,1]$, a non-trivial threshold $t(w,r)$ satisfying
\[
\Prob(w(H) < t(w,r)) \leq r,
\]
where $H$ is a uniformly random Hamilton cycle from $\mathcal{H}_n$. We achieve this using martingale concentration inequalities, for which we now introduce the necessary setup, following McDiarmid \cite{McD}.

Let $(\Omega, \mathcal{F}, \Prob)$ be a finite probability space with $(\Omega, \emptyset)=\mathcal{F}_0 \subseteq \mathcal{F}_1 \subseteq \cdots \subseteq \mathcal{F}_n = \mathcal{F}$ a filtration of $\mathcal{F}$. A martingale is a sequence of finite real-valued random variables $X_0, X_1, \ldots, X_n$ such that $X_i$ is $\mathcal{F}_i$-measurable and $\E(X_i \mid \mathcal{F}_{i-1}) = X_{i-1}$ for all $i = 1, \ldots, n$. Note that for any real-valued $\mathcal{F}$-measurable random variable $X$, the sequence of random variables given by $X_i:=\E(X \mid \mathcal{F}_i)$ is a martingale.  
The \emph{difference sequence} $Y_1, \ldots, Y_n$ of a martingale $X_0,X_1, \ldots, X_n$ is the sequence of random variables  given by $Y_i := X_i - X_{i-1}$. The \emph{predictable quadratic variation} of a martingale is defined to be the random variable $W$ given by
\[
W := \sum_{i=1}^n \E(Y_i^2 \mid \mathcal{F}_{i-1}).
\]

We shall use the following variant of Freedman's inequality; see e.g.\ Theorem~3.15 in the survey \cite{McD} of McDiarmid.

\begin{thm} \label{thm:Freedman}
Let $X_0, X_1, \ldots, X_n$ be a martingale with difference sequence $Y_1, \ldots, Y_n$ and predictable quadratic variation $W$. If there exist constants $R$ and $\sigma^2$ such that $|Y_i| \leq R$ for all $i$ and $W \leq \sigma^2$, then %
\COMMENT{McDiarmid only states the one-sided inequality (in different notation); this one is the obvious two-sided one}
\[
\Prob(|X_n - X_0| \geq  t) \leq 2\exp\left( - \frac{t^2/2}{\sigma^2 + Rt/3} \right).
\] 
\end{thm}

In our setting, we work with the probability space $(\mathcal{\tilde{H}}_n, \mathcal{F}, \Prob)$, where $\Prob$ is the uniform distribution on $\mathcal{\tilde{H}}_n$ (and so $\mathcal{F}$ is the power set of $\mathcal{\tilde{H}}_n$). Given an instance $(n,w)$ of TSP, we aim to study the random variable $X: \mathcal{\tilde{H}}_n \rightarrow \mathbb{R}$ given by $X(H):= w(H)$ 
(here $H$ is directed but we interpret $w(H)$ in the obvious way to mean the sum of the weights of the undirected edges of $H$). Thus $X = w(H)$, where $H$ is a uniformly random Hamilton cycle of $\mathcal{\tilde{H}}_n$.

We define a filtration $\mathcal{F}_0 \subseteq \cdots \subseteq \mathcal{F}_{n-1} = \mathcal{F}$, where $\mathcal{F}_k$ is given by fixing the first $k$ vertices of Hamilton cycles. Let us define this more formally. We start by fixing a distinguished vertex $v_0$ of $K_n$ that represents the start of our Hamilton cycle (we will say more later on how $v_0$ should be chosen).
Define $\seq(V_n,k)$ to be the set of sequences $(v_1, \ldots, v_k)$ of length $k$ where $v_1, \ldots, v_k$ are distinct vertices from $V_n \setminus \{v_0 \}$. For  $s = (v_1, \ldots, v_k) \in \seq(V_n,k)$, define $\tilde{\mathcal{H}}_n(s)$ to be the set of Hamilton cycles whose first $k$ vertices after $v_0$ are $v_1, \ldots, v_k$ in that order. Then $\mathcal{F}_k$ is the $\sigma$-field generated by $\{ \mathcal{H}_n(s): s \in \seq(V_n,k) \}$, and it is clear that $\mathcal{F}_0 \subseteq \cdots \subseteq \mathcal{F}_{n-1} = \mathcal{F}$ is a filtration of $\mathcal{F}$.

Thus we obtain a martingale $X_0, \ldots, X_{n-1}$ by setting $X_i := \E(X \mid \mathcal{F}_i)$. We call this the \emph{Hamilton martingale for $(n,w)$}. Note that $X_{n-2}=X_{n-1}=X$ (this is because knowing the order of the first $n-1$ vertices of an $n$-vertex Hamilton cycle determines it completely).

Let us return to the question of how the distinguished vertex $v_0$ should be chosen. Given our instance $(n,w)$ of TSP, by simple averaging we can and shall choose $v_0$ to be a vertex such that 
\begin{equation} \label{eq:v0}
\sum_{v \in V_n \setminus \{ v_0 \}}|w(v_0v)| \leq \frac{2w[K_n]}{n}.
\end{equation}
We require such a choice of $v_0$ in order to effectively
bound the difference sequence and predictable quadratic variation of Hamilton martingales in the following lemma.

\begin{lem} \label{lem:martbounds}
Let $(n,w)$ be an instance of TSP and assume that $w[K_n] = d \binom{n}{2}$ for some $d \in [0,1]$.
Let $X_0, X_1, \ldots, X_{n-1}$ be the Hamilton martingale for $(n,w)$, let $Y_1, \ldots, Y_{n-1}$ be its difference sequence and let $W$ be its predictable quadratic variation. Then we have the following uniform bounds 
$|Y_i| \leq 6$ for all $i$ and $W \leq 60( \sqrt{d}n + 1)$.
\end{lem}

Before we prove the lemma, we prove a few basic properties of $w$.

\begin{prop}
\label{pr:b1}
Suppose $(n,w)$ is an instance of TSP with $w[K_n] = d \binom{n}{2}$ for some $d \in [0,1]$. For every $A \subseteq V_n$ with $|A|=r$, we have 
$w[A^{(2)}] \leq d_r \binom{r}{2}$, where $d_r := \min\{1, d\binom{n}{2}/\binom{r}{2}\}$. Furthermore
\[
\sum_{r=2}^{n-1}d_r \leq 2\sqrt{d}n + 1.
\]
\end{prop}
\begin{proof}
For $(n,w)$ and $A$ as in the statement of the proposition, we clearly have that $w[A^{(2)}] \leq \min\{\binom{r}{2}, w[K_n]\} = d_r \binom{r}{2}$.
Furthermore, we have
\begin{align*}
\sum_{r=2}^{n-1} d_r 
&\leq \sum_{r=2}^{\lceil \sqrt{d}n \rceil +1}1 
+  \sum_{r =\lceil \sqrt{d}n \rceil +1}^{n-1} d\binom{n}{2}/\binom{r}{2} \\
&=  \lceil \sqrt{d}n \rceil  + 
dn(n-1)\sum_{r =\lceil \sqrt{d}n \rceil +1}^{n-1}\left(\frac{1}{r-1} - \frac{1}{r}\right) \\
&\leq \lceil \sqrt{d}n \rceil  + \frac{dn(n-1)}{\lceil \sqrt{d}n \rceil} \leq 2\sqrt{d}n + 1.
\end{align*}
\end{proof}

Given an instance $(n,w)$ of TSP and a Hamilton cycle $H \in \mathcal{\tilde{H}}$ whose vertices are ordered $v_0, v_1, \ldots, v_{n-1}$, we write
\[
w_H^+(v_i) := \sum_{i+1 \leq j \leq n-1}|w(v_iv_j)|.
\] 

\begin{prop}
\label{pr:b2}
Suppose $(n,w)$ is an instance of TSP with $w[K_n] = d \binom{n}{2}$ for some $d \in [0,1]$ and let $H \in \mathcal{\tilde{H}}_n$. Then
\[
\sum_{i=1}^{n-2} \frac{w_H^+(v_{i-1})}{n-i} \leq \sqrt{d}n + 2.
\]
\end{prop}
\begin{proof}
Let $(n,w)$ and $H$ be as in the statement of the proposition and let $v_0, v_1, \ldots, v_{n-1}$ be the ordering of vertices given by $H$.
Let $e_1, e_2, \ldots, e_{\binom{n}{2}}$ be the lexicographic ordering on the edges of $K_n$ induced by the vertex ordering of $H$. If $e_i = v_jv_k$ with $j<k$, then set $\lambda_i := 1/(n-j-1)$. Thus
\[
\sum_{i=1}^{n-2} \frac{w_H^+(v_{i-1})}{n-i}
\leq \sum_{i=1}^{n-1} \frac{w_H^+(v_{i-1})}{n-i}
 = \sum_{i=1}^{\binom{n}{2}} \lambda_i |w(e_i)|.
\]
Note that the $\lambda_i$  form an increasing sequence and so $\sum \lambda_i |w(e_i)|$ is maximised (subject to the constraints that $w(e) \in [-1,1]$ for all $e$ and $w[K_n] = d\binom{n}{2}$) by maximising the weights of edges at the end of the lexicographic order.
Therefore we obtain an overestimate of $\sum \lambda_i |w(e_i)|$ by assigning a weight of $1$ to the last $\binom{r}{2}$ edges (and $0$ to all other edges), where $r$ is chosen such that $\binom{r}{2} \geq d \binom{n}{2}$. The last inequality is satisfied  by taking $r := \lceil \sqrt{d}n \rceil + 1$, and so
\begin{align*}
\sum_{i=1}^{n-2} \frac{w_H^+(v_{i-1})}{n-i}
\leq \sum_{i=1}^{n-1} \frac{w_H^+(v_{i-1})}{n-i} 
\leq \sum_{i=n-r}^{n-1} \frac{n-i}{n-i} = r = \lceil \sqrt{d}n \rceil + 1. 
\end{align*}
\end{proof}

Before proving Lemma~\ref{lem:martbounds} we introduce some further notation. For any $H \in \tilde{H}_n$, define $s_i(H) = (v_1, \ldots, v_i) \in \seq(V_n,i)$, where $v_1, \ldots, v_i$ are the first $i$ vertices (after $v_0$) of $H$, in that order.

If $Z$ is a random variable on $(\mathcal{\tilde{H}}_n, \mathcal{F}, \Prob)$, then as usual, we write $Z(H)$ for the value of $Z$ at $H \in \mathcal{\tilde{H}}_n$. For $s = (v_1, \ldots, v_k) \in \seq(V_n,k)$, we write $\E(Z \mid s) = \E(Z \mid v_1, \ldots, v_k)$ to mean the expected value of $Z$ given that the first $k$ vertices (after $v_0$) of our uniformly random Hamilton cycle from $\tilde{H}_n$ are $v_1, \ldots, v_k$ in that order. Thus we have $\E(Z \mid \mathcal{F}_i)(H) = \E(Z \mid s_i(H))$ for all $H \in \mathcal{\tilde{H}}_n$ and in particular, we have $X_i(H) = \E(X \mid s_i(H))$. %
\COMMENT{You may find this notationally heavy. I'm essentially explaining two different but equivalent notations because in parts one form is more intuitive and in other parts the other is more intuitive.}

Finally, for any $s = (v_1, \ldots, v_i) \in \seq(V_n,i)$, 
write $V(s):=\{v_0, v_1, \ldots, v_i\}$ 
and $\overline{V}(s):=V_n \setminus V(s)$.

\removelastskip\penalty55\medskip\noindent{\bf Proof of Lemma~\ref{lem:martbounds}}.%
\COMMENT{most technical part of the paper}
Fix a Hamilton cycle $\hat{H} \in \mathcal{\tilde{H}}_n$ and let $v_0,v_1, \ldots, v_{n-1}$ be the order of vertices in $\hat{H}$. 
Let $s_k = s_k(\hat{H}) = (v_1, \ldots, v_k)$. Also, let $H$ be a uniformly random Hamilton cycle from  $\mathcal{\tilde{H}}_n$.

We have for each $k=0, \ldots, n-2$ that
\begin{align*}
X_k(\hat{H}) &= \E(X \mid s_k) 
= \sum_{e \in E_n} \Prob(e \in H \mid s_k(H) = s_k)w(e) \\
&= \frac{1}{n-k-1} \sum_{v \in \overline{V}(s_k)}w(v_kv) 
+ \frac{1}{n-k-1} \sum_{v \in \overline{V}(s_k)} w(v_0v) \\
&+ \frac{2}{n-k-1} \sum_{e \in \overline{V}(s_k)^{(2)}}w(e)
+ \sum_{i=1}^k w(v_{i-1}v_i).
\end{align*}%
\COMMENT{To see the $2/(n-k-1)$ term in the third line, note that we want to choose $n-k-2$ edges  in $\overline{V}(s_k)^{(2)}$ and $|\overline{V}(s_k)^{(2)}| = \binom{n-k-1}{2}$.}
Using the above, and after cancellation and collecting terms,\COMMENT{$A_2-A_4-A_5$ are contributions from the first and third sums, $A_1-A_3$ is the contribution from second sum, $A_6$ is contribution from fourth sum.} we obtain, for each $k=1, \ldots, n-2$
\[
Y_k(\hat{H}) = X_k(\hat{H}) - X_{k-1}(\hat{H}) = A_1 + A_2 - A_3 -A_4 -A_5 +A_6,
\]
where $A_i = A_i(\hat{H})$ is given by
\begin{align*}
&A_1 := \left( \frac{1}{n-k-1} - \frac{1}{n-k} \right) 
\sum_{v \in \overline{V}(s_{k-1})}w(v_0v), 
&&A_4 := \frac{1}{n-k-1} 
\sum_{v \in \overline{V}(s_{k})}w(v_kv), \\
&A_2 := 2\left( \frac{1}{n-k-1} - \frac{1}{n-k} \right) 
\sum_{e \in \overline{V}(s_{k-1})^{(2)}}w(e),
&&A_5 := \frac{1}{n-k} 
\sum_{v \in \overline{V}(s_{k-1})}w(v_{k-1}v), \\ 
&A_3 := \frac{1}{n-k-1}w(v_0v_k), 
&&A_6 := w(v_{k-1}v_k).
\end{align*}
In order to bound $|Y_k(\hat{H})|$, we bound each of $|A_1|, \ldots, |A_6|$ in similar ways using the fact that $|w(e)| \leq 1$ for all edges $e$. We have 
\[
|A_1| \leq \frac{1}{(n-k)(n-k-1)}\sum_{v \in \overline{V}(s_{k-1})}|w(v_0v)| \leq \frac{|\overline{V}(s_{k-1})|}{(n-k)(n-k-1)} \leq 1.
\]
We have
\[
|A_2| \leq \frac{2w[\overline{V}(s_{k-1})^{(2)}]}{(n-k)(n-k-1)} \leq
\frac{2|\overline{V}(s_{k-1})^{(2)}|}{(n-k)(n-k-1)} = 1.
\]
Similarly, $|A_3|, |A_4|, |A_5|, |A_6| \leq 1$.%
\COMMENT{Is it worth spelling this out? Clearly $|A_3|,|A_6| \leq 1$. Also $|A_4| \leq |\overline{V}(s_k)|/(n-k-1) \leq 1$ and $|A_5| \leq |\overline{V}(s_{k-1})|/(n-k) \leq 1$.} 
Thus we have $|Y_k(\hat{H})| \leq \sum_{i=1}^6 |A_i| \leq 6$, and since $\hat{H}$ is arbitrary, we have $|Y_k| \leq 6$ as required.

In order to bound $W(\hat{H})$, we must estimate $\E(Y_k^2 \mid \mathcal{F}_{k-1})(\hat{H})$. We have
\begin{align*}
\E(Y_k^2 \mid \mathcal{F}_{k-1})(\hat{H}) &= \E(Y_k^2 \mid s_{k-1}(\hat{H}))
= \E(Y_k^2 \mid v_1, \ldots, v_{k-1})  \\
&= \frac{1}{|\overline{V}(s_{k-1})|}
\sum_{y \in \overline{V}(s_{k-1})} \E(Y_k^2 \mid v_1, \ldots, v_{k-1},y).
\end{align*}
Recall that $Y_k$ is $\mathcal{F}_k$-measurable, so if $H_{k,y} \in \mathcal{\tilde{H}}_n(v_1, \ldots, v_{k-1},y)$, then
\[
\E(Y_k^2 \mid v_1, \ldots, v_{k-1},y) = Y_k^2(H_{k,y}) 
 \leq 6 |Y_k(H_{k,y})|.
\]
Therefore 
\[
\E(Y_k^2 \mid \mathcal{F}_{k-1})(\hat{H}) \leq \frac{6}{n-k}
\sum_{y \in \overline{V}(s_{k-1})} |Y_k(H_{k,y})| 
\leq 
\frac{6}{n-k}
\sum_{i=1}^6\sum_{y \in \overline{V}(s_{k-1})}|A_i(H_{k,y})|.
\]
Setting 
\[
B_i := \frac{1}{n-k} \sum_{y \in \overline{V}(s_{k-1})}|A_i(H_{k,y})|, 
\]
we have 
\begin{align*}
B_1 
&\leq \frac{1}{(n-k)(n-k-1)}\sum_{v \in \overline{V}(s_{k-1})}|w(v_0v)|  \stackrel{(\ref{eq:v0})}{\leq} \frac{2d(n-1)}{(n-k)(n-k-1)}, \\
B_2
&\leq \frac{2}{(n-k)(n-k-1)} \sum_{e \in \overline{V}(s_{k-1})^{(2)}}|w(e)| \leq d_{n-k}, \\
B_3
&\leq \frac{1}{(n-k)(n-k-1)}\sum_{y \in \overline{V}(s_{k-1})}|w(v_0y)| \stackrel{(\ref{eq:v0})}{\leq} \frac{2d(n-1)}{(n-k)(n-k-1)}, \\
B_4
&\leq \frac{1}{n-k} \sum_{y \in \overline{V}(s_{k-1})}\frac{1}{n-k-1} \sum_{v \in \overline{V}(v_1, \ldots v_{k-1},y)}|w(yv)| \\
&= \frac{2}{(n-k)(n-k-1)} \sum_{e \in \overline{V}(s_{k-1})^{(2)}}|w(e)| \leq d_{n-k}, \\
B_5 
&\leq \frac{1}{n-k} \sum_{v \in \overline{V}(s_{k-1})}|w(v_{k-1}v)| = \frac{w^+_{\hat{H}}(v_{k-1})}{n-k}, \\
B_6
&\leq \frac{1}{n-k} \sum_{y \in \overline{V}(s_{k-1})}|w(v_{k-1}y)| 
= \frac{w^+_{\hat{H}}(v_{k-1})}{n-k},
\end{align*}
where the bound the bound for $B_2$ and $B_4$ follows from Proposition~\ref{pr:b1}.
Using these bounds, we obtain
\begin{align*}
\E(Y_k^2 \mid \mathcal{F}_{k-1})(\hat{H})
\leq
6 \left( 
\frac{4d(n-1)}{(n-k)(n-k-1)} + 2d_{n-k} + 
\frac{2}{n-k} w^+_{\hat{H}}(v_{k-1})
 \right).
\end{align*}
Summing this expression over $k=1, \ldots, n-2$,  and using that $\sum_{r=2}^{n-1}\frac{1}{r(r-1)} \leq 1$\COMMENT{$\sum_{r=2}^{n-1}\frac{1}{r(r-1)} = \sum_{r=2}^{n-1}(\frac{1}{r-1} - \frac{1}{r}) = 1 - \frac{1}{n-1}$}, Proposition~\ref{pr:b1}, and Proposition~\ref{pr:b2}, gives
\begin{align*}
W(\hat{H}) 
= \sum_{k=1}^{n-1} \E(Y_k^2 \mid \mathcal{F}_{k-1})(\hat{H})
&= \sum_{k=1}^{n-2} \E(Y_k^2 \mid \mathcal{F}_{k-1})(\hat{H}) \\
&\leq 6 \left( 4d(n-1) + 2(2 \sqrt{d}n+1) + 2(\sqrt{d}n +2) \right) \\
&= 24d(n-1) + 36 \sqrt{d}n + 36 \leq 60(\sqrt{d}n + 1),
\end{align*}
where we have used that $\E(Y_{n-1}^2 \mid \mathcal{F}_{n-2})(\hat{H}) = 0$, which follows since, as we noted earlier, $X_{n-1}=X_{n-2}$.
This proves this lemma, since $\hat{H}$ is arbitrary.
\endproof


\section{$\{0,1\}$-weightings}
\label{sec:alg}

In this section, we provide a polynomial-time TSP-domination algorithm with\COMMENT{TSP9,11,12: reworded} large domination ratio for instances $(n,w)$ of TSP in which $w:E_n \rightarrow \{0,1\}$. We call such instances \emph{$\{0,1\}$-instances} of TSP. In fact our algorithm consists of three separate algorithms, each adapted for different types of $\{0,1\}$-instances.




We begin with the following classical result of Erd\H{o}s and Gallai \cite{ErdGal}%
\COMMENT{See also \cite{Erd} for a clear statement.}
 on the number of edges needed in a graph to guarantee a matching of a given size.

\begin{thm} \label{thm:ErGal}%
Let $n,s$ be positive integers. Then the minimum number of edges in an $n$-vertex graph that forces a matching with $s$ edges is
\[
\max \left\{ \binom{2s-1}{2}, \binom{n}{2} - \binom{n-s+1}{2} \right\} + 1.
\]
\end{thm}

We recast the above result into a statement about 
$\{0,1\}$-instances of TSP and the existence of optimal matchings of $K_n$ with small weight.%
\COMMENT{Conditions in proposition below have changed slightly}

\begin{prop}
\label{prop:matchings}
Let\COMMENT{{\bf n odd}: this result still holds for $n$ odd.} $(n,w)$ be a $\{0,1\}$-instance of TSP with $n \geq 1$ and $w(K_n) = d\binom{n}{2}$ for some $d \in [0,1]$ satisfying $n^{-1} \leq d \leq 1- 4n^{-1}$. Then there exists a optimal matching $M^*$ of $K_n$ such that $w(M^*) \leq f(n,d)$, where 
\[
f(n,d) := 
\begin{cases}
\frac{1}{2}dn - \frac{1}{8}dn + 1 &\text{if } d \leq \frac{9}{25};\\
\frac{1}{2}dn - \frac{1}{8}(1-d)^2n +1 &\text{if } d \geq \frac{9}{25}.
\end{cases}
\]
\end{prop}

We remark that for a random perfect matching $M$ of $K_n$ (where $n$ is even), we have $\E(w(M)) = \frac{1}{2}dn$.%
\COMMENT{{\bf n odd:} slightly smaller if $n$ odd}
 Thus it is instructive to compare the expressions in the statement of Proposition~\ref{prop:matchings} with $\frac{1}{2}dn$. We note in particular that when $d$ is bounded away from $0$ and $1$, we can find a perfect matching whose weight is significantly smaller than that of an average perfect matching, but as $d$ approaches $0$ or $1$, all perfect matchings tend to have roughly the same weight.

\begin{proof}
Let $G$ be the $n$-vertex subgraph of $K_n$ whose edges are the edges of $K_n$ of weight zero; thus $e(G) = (1-d)\binom{n}{2}$. 

If $s = \frac{1}{2} \sqrt{1-d}n$ then it is easy to check that $(1-d)\binom{n}{2} \geq \binom{2s-1}{2} + 1$,%
\COMMENT
{
$\binom{2s-1}{2}+1 = \binom{\sqrt{1-d}n-1}{2}+1 = (1-d)\binom{n}{2} + \frac{1}{2}(1-d)n - \frac{3}{2}\sqrt{1-d}n + 2 \leq (1-d)\binom{n}{2} - \sqrt{1-d}n + 2 \leq (1-d)\binom{n}{2}$; the last inequality holds since $d \leq 1 - 4n^{-1}$ and $n \geq 1$.
}
 and if $s=(1-\sqrt{d})n$ then it is easy to check that $(1-d)\binom{n}{2} \geq \binom{n}{2} - \binom{n-s+1}{2} + 1$.%
\COMMENT
{
$\binom{n}{2} - \binom{n-s+1}{2}+1 = \binom{n}{2} - \binom{\sqrt{d}n+1}{2}+1 = (1-d)\binom{n}{2} - \frac{1}{2}(\sqrt{d}+d)n +1 \leq (1-d)\binom{n}{2}$ since $\sqrt{d}n \geq dn \geq 1$.
}
 Thus Theorem~\ref{thm:ErGal} implies that $G$ has a matching of size at least
\[
g(n,d) := \left\lfloor \min \left\{ \frac{1}{2} \sqrt{1-d}n, (1-\sqrt{d})n \right\} \right\rfloor. 
\] 
Note that if $M$ is any matching of $G$ with $s$ edges, then we can extend $M$ arbitrarily to an optimal matching $M'$ of $K_n$ such that $w(M') \leq (n/2) - s$. Thus there is an optimal matching $M^*$ of $K_n$ with 
\[
w(M^*) \leq n/2 - g(n,d) \leq 1 + \max \left\{ \frac{1}{2}(1-\sqrt{1-d})n, \left( \sqrt{d} - \frac{1}{2} \right)n \right\}.
\] 
It is easy to compute that the maximum above is given by its first term if $d \in [0,9/25]$ and the second when $d \in [9/25,1]$.%
\COMMENT
{
Assume $d \in [0,1]$. Then
$\frac{1}{2}(1-\sqrt{1-d}) \geq \left( \sqrt{d} - \frac{1}{2} \right)$
iff
$\sqrt{d} + \frac{1}{2}\sqrt{1-d} \leq 1$
iff
$d + \frac{1}{4}(1-d) + \sqrt{d(1-d)} \leq 1$
iff
$\sqrt{d(1-d)} \leq \frac{3}{4}(1-d)$
iff 
$d \leq \frac{9}{16}(1-d)$ or $d=1$
iff $d \leq \frac{9}{25}$ or $d=1$ (but at $d=1$ the two terms are equal).
}
 Now using that $\sqrt{1-d} \geq 1- \frac{3}{4}d$ for $d \in [0,9/25]$%
\COMMENT
{
For $d \in [0,9/25]$, we have $\sqrt{1-d} \geq 1- \frac{3}{4}d$ iff $1-d \geq 1 - \frac{3}{2}d + \frac{9}{16}d^2$ iff $\frac{1}{2}d \geq \frac{9}{16}d^2$ iff $d \leq \frac{8}{9}$, which is true.
}
 and that $\sqrt{d} \leq \frac{1}{2} + \frac{1}{2}d - \frac{1}{8}(1-d)^2$ for $d \in [0,1]$,%
\COMMENT
{
Substituting $d=1-x$, we must show for all $x \in [0,1]$ that $\sqrt{1-x} \leq 1 - \frac{1}{2}x - \frac{1}{8}x^2$ iff $1-x \leq 1 - x + \frac{1}{8}x^3 + \frac{1}{64}x^4$, which holds. Note: inequality derived from Taylor expansion.
}
 the proposition easily follows. (The latter inequality can be checked by substituting $1-x$ for $d$, squaring both sides, and rearranging.\COMMENT{TSP10: sentence changed})
\end{proof}

Next\COMMENT{TSP13: paragraph reworded since theorem it refers to has now been moved to the introduction} we combine the various results we have gathered so far to prove Theorem~\ref{thm:unused} by showing that Algorithm~A (see end of Section~\ref{sec:algorithm}) 
has a domination ratio exponentially close to $1$ for $\{0,1\}$-instances of TSP with fixed density, i.e.\ density that is independent of $n$. In fact, we will not make use of Theorem~\ref{thm:unused}
in order to prove Theorem~\ref{thm:main}:%
   \COMMENT{TST14: slightly reformulated}
instead we will make use of a similar and slightly more technical result in which $d$ may depend on $n$.


\removelastskip\penalty55\medskip\noindent{\bf Proof of Theorem~\ref{thm:unused}}.
Given an instance $(n,w)$ as in the statement of the theorem, 
 Proposition~\ref{prop:matchings} implies that there exists an optimal matching $M^*$ of $K_n$ satisfying $w(M^*) \leq \frac{1}{2}dn - \frac{1}{8}\eta^2n + O(1)$. Thus Algorithm~A, which has running time $O(n^5)$ outputs a Hamilton cycle $H^*$ satisfying
\begin{align*}
w(H^*) 
&\leq \left( 1 - \frac{1}{n-2} \right) \left(\frac{1}{2}dn - \frac{1}{8}\eta^2n + O(1) \right) + \frac{1}{n-2}d\binom{n}{2} + \rho(n) \\
&= dn - \frac{1}{8} \eta^2 n + O(1).
\end{align*}

Set $t:= dn - w(H^*) \geq \frac{1}{8} \eta^2 n + O(1)$.
Let $X_0, X_1, \ldots, X_{n-1}$ be the Hamilton martingale for $(n,w)$, so that $X_{n-1} = w(H)$ where $H$ is a uniformly random Hamilton cycle from $\mathcal{\tilde{H}}_n$ and $X_0 = \E(w(H))=dn$. From Theorem~\ref{thm:Freedman} and Lemma~\ref{lem:martbounds}, we have
\begin{align*}
\Prob(w(H) \leq w(H^*)) \leq \Prob(X_{n-1} \leq X_0 - t) 
&\leq 2\exp \left(- \frac{t^2/2}{ \sigma^2 + Rt/3} \right) \\
&\leq O(\exp(-\eta^4n/10^4)),
\end{align*}
where $R=6$ and $\sigma^2 = 60(\sqrt{d}n + 1) \leq 60n + O(1)$.%
\COMMENT
{
We have $\frac{t^2/2}{ \sigma^2 + Rt/3} \geq \frac{\eta^4n^2/128}{60n + n} + O(1) \geq \eta^4n/10^4 + O(1)$.
}
\endproof
We remark that, although we used Theorem~\ref{thm:Freedman} (the variant of Freedman's inequality) in the proof above, Azuma's inequality, which is much simpler to apply, gives the same bounds.%
\COMMENT{TSP10: For Azuma's inequality, we know all martingale differences are bounded by $R=6$, so the sum of the martingale differences is $6n$ or $6(n-1)$ if we ignore the last step. Thus Azuma gives us
\begin{align*}
\Prob(w(H) \leq w(H^*)) \leq \Prob(X_{n-1} \leq X_0 - t) 
&\leq 2\exp \left(- \frac{t^2/2}{ Rn } \right) \\
&\leq O(\exp(-\eta^4n/10^4)).
\end{align*}
}
However, Azuma's inequality is not strong enough to derive our main result, and in particular, it is not strong enough to derive Theorem~\ref{thm:ratio}.

Our next goal is to give a result similar to Theorem~\ref{thm:unused} in which the density of our $\{0,1\}$-instance of TSP can depend on $n$. We begin with the following definition.

\begin{defn}
\label{def:reginstance}
For $d \in [0 , 1]$, $\varepsilon>0$\COMMENT{Actually, $\varepsilon$ needn't necessarily be small; it could even be some small power of $n$ in the next theorem} and an integer $n>\max\{6, \exp(\varepsilon^{-1})\}$,  we call a $\{0,1\}$-instance $(n,w)$ of TSP an \emph{$(n,d,\varepsilon)$-regular instance} if
\begin{enumerate}
\item[(i)] $w(K_n) = d \binom{n}{2}$;
\item[(ii)] there exists an optimal matching $M^*$ of $K_n$ such that
either
$w(M^*) \leq \frac{1}{2}dn - m_{\varepsilon}(n,d)$ or $w(M^*) \leq \frac{1}{2}dn - m_{\varepsilon}(n,1-d)$, where
\[
 m_{\varepsilon}(n,d):= 40 (\varepsilon + \varepsilon^{1/2}) \log n + 40 \varepsilon^{1/2}d^{1/4} \sqrt{n \log n}.
\]
\end{enumerate}
Note that $w(K_n) = w[K_n]$ for $\{0,1\}$-instances of TSP.
\end{defn}

\begin{thm} \label{thm:ratio}
If\COMMENT{{\bf n odd}: result holds for $n$ odd if we assume $n$ large enough as a function of $\varepsilon$} $(n,w)$ is an $(n,d,\varepsilon)$-regular instance of TSP, then Algorithm~A has domination ratio at least $1-2n^{- \varepsilon}$ for $(n,w)$.
\end{thm}
\begin{proof}
Suppose $(n,w)$ is an $(n,d,\varepsilon)$-regular instance of TSP and set $m := m_{\varepsilon}(n,d)$.
Assume first that there exists an optimal matching $M^*$ of $K_n$ with $w(M^*) \leq dn/2 - m$. Then Algorithm~A outputs a Hamilton cycle $H^*$ with
\begin{align*}
w(H^*) &\leq \left(1 - \frac{1}{n-2} \right)(dn/2 - m) + \frac{1}{n-2}d\binom{n}{2} + \rho(n) \\
&= dn - \left( 1 - \frac{1}{n-2} \right)m  + \rho(n)\\
& \leq dn - m/2 ,
\end{align*} 
where the last inequality follows because $n > 6$ and $m>4$ (which follows because $n>\exp(\varepsilon^{-1})$). Set $t := m/2$.%

Let $X_0, X_1, \ldots, X_{n-1}$ be the Hamilton martingale for $(n,w)$, so that $X_{n-1} = w(H)$ where $H$ is a uniformly random Hamilton cycle from $\mathcal{\tilde{H}}_n$ and $X_0 = \E(w(H))=dn$. From Theorem~\ref{thm:Freedman} and Lemma~\ref{lem:martbounds}, we have
\begin{align*}
\Prob(w(H) \leq w(H^*)) \leq \Prob(X_{n-1} \leq X_0 - t) 
\leq 2\exp \left(- \frac{t^2/2}{ \sigma^2 + Rt/3} \right),
\end{align*}
where $R=6$ and $\sigma^2 = 60(\sqrt{d}n + 1)$. We have that
\[
\frac{t^2/2}{ \sigma^2 + Rt/3} \geq 
\min \left \{ \frac{t^2/2}{2 \sigma^2}, \frac{t^2/2}{2 Rt/3} \right\}
=
\min \left\{ \frac{t^2}{4 \sigma^2}, \frac{t}{8} \right\}.
\]
Noting that 
\[
t^2 = \frac{1}{4}(40 (\varepsilon + \varepsilon^{1/2}) \log n + 40 \varepsilon^{1/2}d^{1/4} \sqrt{n \log n})^2 \geq 400 \varepsilon \log n(\sqrt{d}n + 1),
\] %
\COMMENT
{
$t:=t_1 + t_2 := 20(\varepsilon + \varepsilon^{1/2}) \log n + 20\varepsilon^{1/2}d^{1/4} \sqrt{n \log n}$. We have $t^2 \geq t_1^2 + t_2^2$ and $t_1^2 \geq 400 \varepsilon \log n$ and $t_2^2 \geq 400 \varepsilon \log n \sqrt{d} n$.
}%
we have $t^2/4\sigma^2 \geq \varepsilon \log n$. Also $t/8 \geq \varepsilon \log n$, and so we have
\[
\Prob(w(H) \leq w(H^*)) 
\leq 2\exp(- \varepsilon \log n) = 2n^{-\varepsilon},
\]
as required.

Now set $m := m_{\varepsilon}(n,1-d)$ and assume there is an optimal matching  $M^*$ of $K_n$ with $w(M^*) \leq dn/2 - m$. As before Algorithm~A outputs a Hamilton cycle $H^*$ with $w(H^*) \leq dn - m/2$. Again set $t := m/2$.

This time let $X_0, X_1, \ldots, X_{n-1}$ be the Hamilton martingale for $\bar{w} := 1-w$.
From Theorem~\ref{thm:Freedman} and Lemma~\ref{lem:martbounds}, we have
\begin{align*}
\Prob(w(H) \leq w(H^*)) 
= \Prob( \bar{w}(H) \geq \bar{w}(H^*) )
&\leq \Prob(X_{n-1} \geq X_0 + t) \\
&\leq 2\exp \left(- \frac{t^2/2}{ \sigma^2 + Rt/3} \right),
\end{align*}
where $R=6$ and $\sigma^2 = 60(\sqrt{1-d}n + 1)$ (since $\bar{w}(K_n) = (1-d)\binom{n}{2}$). 
Following the same argument as before with $d$ replaced by $1-d$, we obtain
\[
\Prob(w(H) \leq w(H^*)) \leq 2n^{-\varepsilon},
\]
as required.
\COMMENT
{
We have that
\[
\frac{t^2/2}{ \sigma^2 + Rt/3} \geq 
\min \left \{ \frac{t^2/2}{2 \sigma^2}, \frac{t^2/2}{2 Rt/3} \right \}
=
\min \left \{ \frac{t^2}{4 \sigma^2}, \frac{t}{8} \right \}.
\]
Noting that $t^2 = \frac{1}{4}(40 (\varepsilon + \varepsilon^{1/2}) \log n + 40 \varepsilon^{1/2}(1-d)^{1/4} \sqrt{n \log n})^2 \geq 400 \varepsilon \log n(\sqrt{1-d}n + 1)$, we have $t^2/4\sigma^2 \geq \varepsilon \log n$. Also $t/8 \geq \varepsilon \log n$, and so we have
\[
\Prob(w(H) \leq w(H^*)) \leq 2\exp \left( -\frac{t^2/2}{ \sigma^2 + Rt/3} \right) \leq 2\exp(- \varepsilon \log n) = 2n^{-\varepsilon},
\]
}
\end{proof}

\begin{cor}
\label{cor:reg}
For every $\varepsilon>0$ there exists $n_0 \in \mathbb{N}$ such that for all $n > n_0$, the following holds.\COMMENT{TSP10: statement changes so that $n_0$ no longer explicitly given.}
Define $f_{\varepsilon}(n) := 10^4(1+ 2\varepsilon)n^{-2/3}\log n$ and $g_{\varepsilon}(n) := 10^4 (1+ \varepsilon)n^{-2/7}\log n$.
If  
\[
f_{\varepsilon}(n) \leq d \leq 1 - g_{\varepsilon}(n),
\] %
and $(n,w)$ is a $\{0,1\}$-instance of TSP with $w(K_n) = d \binom{n}{2}$ 
then\COMMENT{I have gone for cleaner rather than sharper bounds. It may be worth sharpening the dependence on $\varepsilon$, remembering that $\varepsilon$ does not have to be small.} $(n,w)$ is an $(n,d,\varepsilon)$-regular instance. In particular
 Algorithm~A has domination ratio at least $1-2n^{- \varepsilon}$ for $(n,w)$.
\end{cor}
\begin{proof}
First assume $f_{\varepsilon}(n) \leq d \leq \frac{9}{25}$. From Definition~\ref{def:reginstance}, it is sufficient to exhibit an optimal matching $M^*$ of $K_n$ for which $w(M^*) \leq \frac{1}{2}dn -  m_{\varepsilon}(n,d)$.
By Proposition~\ref{prop:matchings} (note that $d$ satisfies the condition of Proposition~\ref{prop:matchings}),
there exists an optimal matching $M^*$  of $K_n$ such that 
$w(M^*) \leq \frac{1}{2}dn - \frac{1}{16}dn - \frac{1}{16}dn +1$.
Thus we see that $w(M^*) \leq \frac{1}{2}dn - m_{\varepsilon}(n,d)$ if 
\[
\frac{1}{16}dn \geq 40 (\varepsilon + \varepsilon^{\frac{1}{2}}) \log n +1
\:\:\: \text{and} \:\:\:
\frac{1}{16}dn \geq 40\varepsilon^{\frac{1}{2}}d^{\frac{1}{4}} \sqrt{n \log n} .
\]
One can check that both inequalities hold if $d \geq  f_{\varepsilon}(n)$.%
\COMMENT
{ 
We have $\frac{1}{16}dn \geq 40 \varepsilon^{\frac{1}{2}}d^{\frac{1}{4}}\sqrt{n \log n}$ iff $d \geq 640^{\frac{4}{3}} \varepsilon^{\frac{2}{3}} (\log n)^{\frac{2}{3}}n^{-\frac{2}{3}}$, which holds since $d \geq f_{\varepsilon}(n) := 10^4(1+ 2\varepsilon)n^{-2/3}\log n$. 
We also have $\frac{1}{16}dn \geq 40 (\varepsilon + \varepsilon^{\frac{1}{2}}) \log n +1$ iff $d \geq 640 (\varepsilon + \varepsilon^{\frac{1}{2}}) n^{-1} \log n + 16n^{-1}$, which holds by first noting $1+ 2 \varepsilon \geq \varepsilon + \varepsilon^{\frac{1}{2}}$ and using $d \geq f_{\varepsilon}(n) := 10^4(1+ 2\varepsilon)n^{-2/3}\log n$.
} 

Now assume that $\frac{9}{25} \leq d \leq 1 - g_{\varepsilon}(n)$.
From Definition~\ref{def:reginstance}, it is sufficient to exhibit an optimal matching $M^*$ of $K_n$ for which $w(M^*) \leq \frac{1}{2}dn - m_{\varepsilon}(n,1-d)$.
Set $\overline{d} := 1-d$ and note $g_{\varepsilon}(n) \leq \overline{d} \leq \frac{16}{25}$.
By Proposition~\ref{prop:matchings},
there exists an optimal matching $M^*$  of $K_n$ such that $w(M^*) \leq \frac{1}{2}dn - \frac{1}{8}\overline{d}^2 n + 1$.
Thus we see that $w(M^*) \leq \frac{1}{2}dn - m_{\varepsilon}(n,\overline{d})$
if 
\[
\frac{1}{16}\overline{d}^2 n \geq 40 (\varepsilon + \varepsilon^{\frac{1}{2}}) \log n + 1
\:\:\: \text{and} \:\:\:
\frac{1}{16}\overline{d}^2 n \geq 40\varepsilon^{\frac{1}{2}}\overline{d}^{\frac{1}{4}} \sqrt{n \log n}.
\]
One can check that both inequalities hold if $\overline{d} \geq g_{\varepsilon}(n)$.
\COMMENT
{
$\frac{1}{16}\overline{d}^2 n \geq 40\varepsilon^{\frac{1}{2}}\overline{d}^{\frac{1}{4}} \sqrt{n \log n}$ iff $\overline{d} \geq 640^{\frac{4}{7}}\varepsilon^{\frac{2}{7}} n^{-\frac{2}{7}} (\log n)^{\frac{2}{7}}$, which holds since $\overline{d} \geq g_{\varepsilon}(n) := 10^4 (1+ \varepsilon)n^{-2/7}\log n$.
We have
$40 (\varepsilon + \varepsilon^{\frac{1}{2}}) \log n + 1 \leq
 40 (1 + 2\varepsilon) \log n + 1 \leq 40 (2 + 2\varepsilon)
 \log n \leq  \frac{1}{16}\overline{d}^2 n$, where the first
 inequality holds because $1 + 2\varepsilon \geq \varepsilon + \varepsilon^{\frac{1}{2}}$, and the second inequality holds iff $\overline{d} \geq 640^{\frac{1}{2}}(2+2\varepsilon)^{\frac{1}{2}}(\log n)^{\frac{1}{2}}n^{-\frac{1}{2}}$, which holds since $\overline{d} \geq g_{\varepsilon}(n) := 10^4 (1+ \varepsilon)n^{-2/7}\log n$.
}  
\end{proof}

We have seen in the previous corollary and theorem that Algorithm~A has a large domination ratio for $\{0,1\}$-instances of TSP when $d$ is bounded away from $0$ and $1$, or when there exists an optimal matching of $K_n$ whose weight is significantly smaller than the average weight of an optimal matching. Theorem~\ref{thm:dense} and Theorem~\ref{thm:sparse} give polynomial-time TSP-domination algorithms for all remaining $\{0,1\}$-instances. Before we can prove these, we require some preliminary results. 

Our first lemma is a structural stability result.
Suppose $G \subseteq K_n$ is a graph with $d\binom{n}{2}$ edges, and let $w$ be a weighting of $K_n$ such that $w(e)=1$ if $e \in E(G)$ and $w(e)=0$ otherwise. Then for a random optimal matching $M$ of $K_n$, we have $\E(w(M)) = dn/2$ if $n$ is even and $\E(w(M)) = d(n-1)/2$ if $n$ is odd; this shows that $G$ has a matching of size at least $d(n-1)/2$. We cannot improve much on this if $d$ is close to zero: consider the graph $H$ with a small set $A \subseteq V(H)$ such that $E(H)$ consists of all edges incident to $A$.
The next lemma says that any graph whose largest matching is only slightly larger than $dn/2$ must be similar to the graph $H$ described above. 


\begin{lem}
\label{lem:vxcover}
Fix an integer $s \geq 2$ and let $G$ be an $n$-vertex graph with $e(G) = d \binom{n}{2}$ for some $d \in (0, (4s)^{-1}]$. If the largest matching of $G$ has at most $\frac{1}{2}dn + r$ edges for some positive integer $r<n/(8s)$, then $G$ has a vertex cover $D$ with
\[
 |D| \leq \frac{1}{2}d(n+2) + sd^2(n+1) +(7s+2)r. 
\]
Let 
\[
 S:= \{v \in D: \deg_G(v) \leq s|D| \}. 
\]
Then $|S| \leq 2sd^2(n+1) + (14s+4)r + 3d$.
Furthermore, $D$ and $S$ can be found in $O(n^4)$-time.%
\COMMENT{Is it likely that something similar to this may have been proved somewhere?}
\COMMENT{TSP10: numbers have changes in the statement and proof of this lemma, with consequences elsewhere.}
\end{lem}
\begin{proof}
Let $M^*$ be any matching of $G$ of maximum size and set $m:=e(M^*)$; hence $\frac{1}{2}d(n-1) \leq m \leq \frac{1}{2}dn + r$ (the lower bound follows from the remarks before the statement of the lemma). Let $U$ be the set of vertices of $G$ incident to edges of $M^*$; thus $d(n-1) \leq |U| = 2m \leq dn + 2r$ and $U$ is a vertex cover (by the maximality of $M^*$).
Let 
\[
A := \{ u \in U: \deg_G(u) \leq 2sm  \} \subseteq U.
\]
We bound the size of $A$ as follows. We have
\begin{align*}
d \binom{n}{2} &= e(G) \leq \sum_{u \in U} \deg_G(u) \leq |A|2sm + |U \setminus A|n 
&= 2mn - |A|(n - 2sm) \\
&\leq (dn+2r)n - |A|(n - sdn - 2sr).
\end{align*}
Rearranging gives
\begin{align*}
|A| &\leq (1- sd - 2sr/n)^{-1}\left( \frac{1}{2}dn + \frac{1}{2}d + 2r \right) \\
&\leq (1 + 2sd + 4sr/n) \left( \frac{1}{2}dn + \frac{1}{2}d + 2r \right);
\end{align*}
the last inequality follows by noting that $(1-x)^{-1} \leq 1 + 2x$ for $x \leq \frac{1}{2}$ and that $sd + 2sr/n \leq \frac{1}{2}$ (by our choices of $s,r,d$). Expanding the expression above, and using that $r \leq n/(8s)$ gives
\begin{align*}
|A| &\leq \frac{1}{2}dn + sd^2n + 6sdr + \frac{1}{2}d + sd^2 + 2r + \frac{2sdr + 8sr^2}{n} \\
&\leq \frac{1}{2}d(n+2) + sd^2(n+1) + (6sd+2)r + \frac{8sr^2}{n} \\
&\leq \frac{1}{2}d(n+2) + sd^2(n+1) + (7s+2)r. 
\end{align*}
Let $B := U \setminus A$. 

\medskip

\noindent
\textbf{Claim~1.} \emph{No edge of $M^*$ lies in $B$.}

\smallskip

\noindent Indeed suppose $e = xy$ is an edge of $M^*$ with $x,y \in B$. Then $\deg_G(x), \deg_G(y) > 2sm  \geq |U|+2$ (since $s \geq 2$). Hence there exist distinct $x',y' \in V \setminus U$ such that $xx', yy' \in E(G)$. Replacing $e$ with the two edges $xx', yy'$ in $M^*$ gives a larger matching, contradicting the choice of $M^*$. This proves the claim.

\medskip

\noindent
Let $C: = \{u \in A: uv \in E(M^*), v \in B \}$. By Claim~1, we have $|B|=|C|$ and $B,C$ are disjoint.

\medskip

\noindent
\textbf{Claim~2.} \emph{There are no edges of $G$ between  $C$ and $V \setminus U$, i.e. $E_G(C, V \setminus U) = \emptyset$.}

\smallskip

\noindent
Indeed suppose $e=xy \in E(G)$ with $x \in C$ and $y \in V \setminus U$. By definition of $C$, there exists $z \in B$ such that $xz \in E(M^*)$. By the definition of $B$, $\deg_G(z) > 2sm \geq |U| + 2$, and so we can find $z' \in V\setminus U$ distinct from $y$ such that $zz' \in E(G)$. Then we can replace $xz$ with the two edges $xy, zz'$ in $M^*$ to obtain a larger matching, contradicting the choice of $M^*$. This proves the claim.

\medskip

\noindent
\textbf{Claim~3.} \emph{$C$ is an independent set.}

\smallskip

\noindent
Indeed, suppose $e=xy \in E(G)$ with $x,y \in C$. By definition of $C$, there exist $x',y' \in B$ such that $xx',yy' \in E(M^*)$. By definition of $B$, $\deg_G(x'), \deg_G(y') > 2sm \geq |U|+2$, and so there exist distinct $x'',y'' \in V \setminus U$ such that $x'x'', y'y'' \in E(G)$. Now replace $xx', yy'$ with $xy, x'x'', y'y''$ in $M^*$ to obtain a larger matching, contradicting the choice of $M^*$. This proves the claim.

\medskip

\noindent

Set $D:= U \setminus C$. First we check that $D$ is a vertex cover; indeed note that $V \setminus D = (V \setminus U) \cup C$. But $E_G(V \setminus U)$, $E_G(V \setminus U, C)$, and $E_G(C)$ are all empty using respectively the fact that $U$ is a vertex cover, Claim~2, and Claim~3. Also, we have
\begin{align*}
|D| &= |U|-|C| = |U|-|B| = |U \setminus B| = |A| \\
&\leq \frac{1}{2}d(n+2) + sd^2(n+1) + (7s+2)r.
\end{align*}
Finally, let
\begin{align*}
S := \{v \in D: \deg_G(v) \leq s|D| \} \subseteq \{v \in D: \deg_G(v) \leq s|U| \} = A \cap D.
\end{align*}
So, we have
\begin{align*}
|S| &\leq |A \cap D| = |A| + |D| -|A \cup D| = |A| + |D| - |U| \\
&\leq 2\left( \frac{1}{2}d(n+2) + sd^2(n+1) + (7s+2)r \right) - d(n-1) \\
&= 2sd^2(n+1) + (14s+4)r + 3d.
\end{align*}

Note that $M^*$ can be found in $O(n^4)$-time by suitably adapting the minimum weight optimal matching algorithm (Theorem~\ref{thm:matalg}). From this, $A$, $B$, $C$, $D$, and $S$ can all be constructed in $O(n^2)$-time, as required. 
\end{proof}

Next we give a polynomial-time algorithm for finding a maximum \emph{double matching} in a bipartite graph: it is a simple consequence of the minimum weight optimal matching algorithm from Theorem~\ref{thm:matalg}. Given a bipartite graph $G=(V,E)$ with vertex classes $A$ and $B$, a \emph{double matching} of $G$ from $A$ to $B$ is a subgraph $M = (V, E')$ of $G$ in which $\deg_M(v) \leq 2$ for all $v \in A$ and $\deg_M(v) \leq 1$ for all $v \in B$.

\begin{thm}
\label{thm:mdmt}
There exists an $O(n^4)$-time algorithm which, given an $n$-vertex bipartite graph $G$ with vertex classes $A$ and $B$ as input, outputs a double matching $M^*$ of $G$ from $A$ to $B$ with 
\[
e(M^*) = \max_M e(M),
\]
where the maximum is over all double matchings of $G$ from $A$ to $B$. We call this algorithm the \emph{maximum double matching algorithm}.
\end{thm}
\begin{proof}
Recall that the minimum weight optimal matching algorithm of Theorem~\ref{thm:matalg} immediately gives an $O(n^4)$-time algorithm for finding a maximum matching in a graph\COMMENT{which is in fact the better-known algorithm}.

Given $G$, define a new bipartite graph $G'$ by creating a  copy $a'$ of each vertex $a \in A$
so that $a'$ and $a$ have the same neighbourhood. Note that every matching $M'$ of $G'$ corresponds to a double matching $M$ of $G$ by identifying each $a \in A$ with its copy $a'$ and retaining all the edges of $M'$; hence $e(M) = e(M')$. 

Therefore finding a double matching in $G$ from $A$ to $B$ of maximum size is equivalent to finding a matching of $G'$ of maximum size, and we can use the minimum weight optimal matching algorithm of Theorem~\ref{thm:matalg} to find such a matching.
\end{proof}

Next we prove a lemma that says that for a small subset $S$ of vertices of $K_n$, almost all Hamilton cycles avoid $S$ `as much as possible'.

\begin{lem}
\label{lem:S}
Fix $\varepsilon \in (0, 1/2)$ and $n \geq 4$. Let $S\subseteq V_n$ be a subset of the vertices of the $n$-vertex complete graph $K_n=(V_n,E_n)$ with $|S| \leq n^{\frac{1}{2} - \varepsilon}$. Let $H \in \mathcal{H}_n$ be a uniformly random Hamilton cycle of $K_n$.
Let $\mathcal{E} := \mathcal{E}_1 \cap \mathcal{E}_2$, where $\mathcal{E}_1$ is the event that $H$ uses no edge of $S$ and $\mathcal{E}_2$ is the event that $e_H(v,S) \leq 1$ for all $v \in V_n \setminus S$.
Then
\[
\Prob(\mathcal{E}) \geq 1 - 6n^{-2 \varepsilon}.
\] 
\end{lem}
\begin{proof}
For each $e \in E_n$, we have $\Prob(e \in E(H))= 2/(n-1)$, and so 
\begin{align}
\label{eq:1}
\Prob(\overline{\mathcal{E}_1}) 
&= \Prob(|E(H) \cap S^{(2)}| \geq 1)  
\leq \E(|E(H) \cap S^{(2)}|) \\
&= \binom{|S|}{2}\frac{2}{n-1} \leq \frac{1}{2}n^{1-2\varepsilon} \frac{2}{n-1} \leq 2n^{-2 \varepsilon}. 
\nonumber
\end{align}
Let $R$ be the random variable counting the number of vertices $v \in V_n \setminus S$ such that $e_H(v,S)=2$.
We have
\begin{align}
\label{eq:2}
\Prob(\overline{\mathcal{E}_2}) 
&\leq \E(R) = |V_n \setminus S| \Prob(e_H(v,S)=2)  \\
&= |V_n \setminus S| \frac{|S|(|S|-1)}{(n-1)(n-2)} \leq \frac{n^{2- 2\varepsilon}}{(n-2)^2} \leq 4 n^{-2 \varepsilon}.
\nonumber    
\end{align}
The lemma follows immediately from (\ref{eq:1}) and (\ref{eq:2}).
\end{proof}

We are now ready to present a polynomial-time TSP-domination algorithm which has a large domination ratio for $\{0,1\}$-instances of TSP in which $w(K_n) = d\binom{n}{2}$ with $d$  close to $1$ and in which all optimal matchings have weight close to $dn/2$. Let us first formalise this in a definition.

\begin{defn}
\label{def:instance}
For $n \geq 4$ a positive integer, $r>0$, $d \in [0,1]$, and $\varepsilon \in (0 , 1/2)$ we call a $\{0,1\}$-instance $(n,w)$ of TSP an \emph{$(n,d,r,\varepsilon)$-dense instance} if
\begin{enumerate}
\item[(i)] $w(K_n) = d \binom{n}{2}$ with $1-d \leq 1/12$;
\item[(ii)] $w(M) \geq \frac{1}{2}dn - r$ for all optimal matchings $M$ of $K_n$; and
\item[(iii)]$6\overline{d}^2(n+1) + 46r + 3\overline{d} \leq n^{\frac{1}{2} - \varepsilon}$, where $\overline{d}:=1-d$.%
\COMMENT{This condition implies $r<n/24$, needed to apply Lemma~\ref{lem:vxcover}}
\end{enumerate}
\end{defn}

\begin{thm}
\label{thm:dense}
There exists an $O(n^4)$-time TSP-domination algorithm which has domination ratio at least $1-6n^{-2\varepsilon}$ for all $(n,d,r,\varepsilon)$-dense instances $(n,w)$ of TSP. We call this  Algorithm~B.
\end{thm}
\begin{proof}
We begin by describing Algorithm~B in steps with the running time of each step in brackets. We then explain each step. Let $(n,w)$ be a $(n,d,r,\varepsilon)$-dense instance of TSP and let $G=(V_n, E) \subseteq K_n$ be the graph in which we have $e \in E$ if and only if $w(e)=0$. We set $V:=V_n$ in order to reduce notational clutter.

\begin{enumerate}
\item[1.] Construct $G$ and note that $e(G) = \overline{d} \binom{n}{2}$, where $\overline{d} := 1-d$. ($O(n^2)$ time)
\item[2.] Using the algorithm of Lemma~\ref{lem:vxcover} (with $s=3$), find a vertex cover $D$ for $G$ such that 
$|D| \leq \frac{1}{2}\overline{d}(n+2) + 3\overline{d}^2(n+1) + 23r$.
($O(n^4)$ time)
\item[3.] Construct the bipartite graph $G':=G[D, V \setminus D]$. ($O(n^2)$ time)
\item[4.] Apply the maximum double matching algorithm (Theorem~\ref{thm:mdmt}) to $G'$ to obtain a double matching $M^*$ of $G'$ from $D$ to $V \setminus D$ of maximum size. ($O(n^4)$-time).
\item[5.] Extend $M^*$ arbitrarily to a Hamilton cycle $H^*$ of $K_n$ (i.e.\ choose $H^*$ to be any Hamilton cycle of $K_n$ that includes all the edges of $M^*$). ($O(n)$ time)
\end{enumerate}

Altogether, the running time of the algorithm is $O(n^4)$. Note that in Step 2, the conditions of Lemma~\ref{lem:vxcover} are satisfied by $(n,d,r, \varepsilon)$-dense instances of TSP. Also the algorithm of Lemma~\ref{lem:vxcover} gives us the set
\[
S = \{v \in D: \deg_G(v) \leq 3|D| \},
\]
where $|S| \leq 6\overline{d}^2(n+1) + 46r + 3\overline{d} \leq n^{\frac{1}{2} - \varepsilon}$, which we shall require in the analysis of Algorithm~B.


Next we show that this algorithm has domination ratio at least $1-6n^{-2\varepsilon}$ for $(n,w)$ by showing that for a uniformly random Hamilton cycle $H \in \mathcal{H}_n$, we have
$\Prob(w(H) < w(H^*)) \leq 6n^{-2 \varepsilon}$.
Since $|S| \leq n^{\frac{1}{2} - \varepsilon}$, we can apply Lemma~\ref{lem:S} to conclude that 
\[
\Prob(\mathcal{E}) \geq 1 - 6n^{- 2\varepsilon},
\]
where $\mathcal{E} := \mathcal{E}_1 \cap \mathcal{E}_2$ with  $\mathcal{E}_1$ being the event that $H$ uses no edge of $S$ and $\mathcal{E}_2$ being the event that $e_H(v,S) \leq 1$ for all $v \in V \setminus S$.
We show that if $H$ is any Hamilton cycle in $\mathcal{E} \subseteq \mathcal{H}_n$ then $w(H^*) \leq w(H)$, i.e.\
\begin{equation} \label{eqn:cond}
\Prob(w(H) < w(H^*) \mid \mathcal{E}) = 0.
\end{equation}
Assuming (\ref{eqn:cond}), we can prove the theorem since
\begin{align*}
\Prob(w(H) < w(H^*)) 
&= \Prob(w(H) < w(H^*) \mid \mathcal{E})\Prob(\mathcal{E}) + 
\Prob(w(H) < w(H^*) \mid \overline{\mathcal{E}})\Prob(\overline{\mathcal{E}}) \\
&\leq 0 + \Prob(\overline{\mathcal{E}}) \leq 6n^{-2\varepsilon},
\end{align*}
as required. 

In order to verify (\ref{eqn:cond}), it is sufficient to show that $|E(H^*) \cap E(G)| \geq |E(H) \cap E(G)|$ for all $H \in \mathcal{E}$.
\medskip

\noindent
\textbf{Claim.} \emph{$L:= E(M^*) \cap E_G(S, V \setminus D)$ is a maximum-sized double matching of $G[S, V \setminus D]$ from $S$ to $V \setminus D$\COMMENT{technically, $L$ forms the edges of a double matching}.}

\smallskip

\noindent 
Suppose, for a contradiction, that there is some double matching $L'$ of $G[S, V \setminus D]$ from $S$ to $V \setminus D$ with more edges than $L$. Then since every vertex of $D \setminus S$ has degree at least $3|D|$ in $G$, we can greedily extend $L'$
to a double matching $M^{**}$ of $G[D, V \setminus D]$ such that in $M^{**}$, every vertex of $D \setminus S$ has degree $2$. Hence
\[
e(M^{**}) = e(L') + 2|D \setminus S| > e(L) + 2|D \setminus S| \geq e(M^*),
\]
contradicting the maximality of $M^*$. This proves the claim.
\medskip

\noindent
Now for $H \in \mathcal{E} \subseteq \mathcal{H}_n$, we have that $E(H) \cap E_G(S, V \setminus D)$ is a double matching of $G[S, V \setminus D]$ from $S$ to $V \setminus D$. Hence by the claim and the definition of $H^*$
\begin{equation} \label{eqn:DM}
|E(H) \cap E_G(S, V \setminus D)| \leq 
|E(M^*) \cap E_G(S, V \setminus D)| \leq
|E(H^*) \cap E_G(S, V \setminus D)|.
\end{equation}
Writing $D_S:=D \setminus S$, we have that 
\begin{align*}
E(H) \cap E(G) 
&= \;[E(H)\cap E_G(D_S, V)] \; \cup \; 
[E(H) \cap E_G(V \setminus D_S)] \\
&=\;[E(H) \cap E_G(D_S, V)] \; \cup \; 
[E(H) \cap E_G(S, V \setminus D)]. 
\end{align*}
Now we see that 
\[
|E(H) \cap E(G)| \leq 2|D_S| + |E(H) \cap E_G(S, V \setminus D)|.
\]
Also note that $e_{H^*}(v, V \setminus D) = e_{M^*}(v, V \setminus D)=2$ for all $v \in D_S$ since $M^*$ is maximal and vertices in $D_S$ have degree at least $3|D|$ in $G$. Hence
\[
|E(H^*) \cap E(G)| \geq 2|D_S| + |E(H^*) \cap E_G(S, V \setminus D)|.
\]
But now (\ref{eqn:DM}) implies $|E(H^*)\cap E(G)| \geq |E(H) \cap E(G)|$, as required. 
\end{proof}

Finally we provide a TSP-domination algorithm that has a large domination ratio for sparse $\{0,1\}$-instances of TSP. Before we do this, we require one subroutine for our algorithm.

\begin{lem}
\label{lem:dirac}
There exists an $O(n^3)$-time algorithm which, given an $n$-vertex graph $G$ with $\delta(G) \geq n/2 + \frac{3}{2}k$ and $k$ independent edges $e_1, \ldots, e_k$, outputs a Hamilton cycle $H^*$ of $G$ such that $e_1, \ldots, e_k \in E(H^*)$. 
\end{lem}
The proof of this lemma is a simple adaptation of the proof of Dirac's theorem, but we provide the details for completeness. 
A threshold of $n/2 + k/2$ was proved by P{\'o}sa \cite{Posa}, but it is not clear whether the proof is algorithmic. \COMMENT{TSP10: sentence above changed + reference added}
\begin{proof}
Let $e_i = x_iy_i$, and set $S := \{x_1, \ldots, x_k, y_1, \ldots, y_k \}$ so that $|S|=2k$. By the degree condition of $G$, every pair of vertices have at least $3k$ common neighbours and so $G$ is $3k$-connected. In particular, we can pick distinct vertices $z_1, \ldots, z_{k-1} \in V(G) \setminus S$ such that $z_i$ is a common neighbour of $y_i$ and $x_{i+1}$. Hence $P := x_1y_1z_1x_2y_2z_2 \cdots z_{k-1}x_ky_k$ is a path of $G$ containing all the edges $e_1, \ldots, e_k$. This can be found in $O(nk)$ time.

We say any path or cycle $Q$ is \emph{good} if it contains all the edges $e_1, \ldots, e_k$; in particular $P$ is good. We show that if $Q$ is a good path or cycle with $|V(Q)|<n$ or $e(Q)<n$, then we can extend $Q$ to a good path or cycle $Q'$ such that either $|V(Q')|>|V(Q)|$ or $e(Q')>e(Q)$. 

If $Q$ is a good cycle with $|V(Q)|<n$, then since $G$ is $3k$-connected, there is some vertex $c \in V(G) \setminus V(Q)$ that is adjacent to some $b \in V(Q)$.
Let $a$ and $a'$ be the two neighbours of $b$ on $Q$. Since $e_1, \ldots, e_k$ are independent edges, either $ab$ or $a'b$, say $ab$, is not amongst $e_1, \ldots, e_k$. Thus we can extend $Q$ to the good path $Q':=aQbc$, and we see $|V(Q')|=|V(Q)|+1$. This takes $O(n^2)$ time.
\COMMENT{TSP9: The paragraph above has changed. The change is also different from the one Deryk suggested.}

Now suppose $Q$ is a good path with end-vertices $a$ and $b$, i.e.\ $Q = aQb$, and one of the end-vertices, say $b$, has a neighbour $c$ in $V(G) \backslash V(Q)$. Then we can extend $Q$ to a good path $Q' := aQbc$, and $|V(Q')|>|V(Q)|$. Checking whether we are in this case and obtaining $Q'$ takes $O(n)$ time.

Finally suppose $Q$ is a good path with end-vertices $a$ and $b$, but where $a$ and $b$ have all their neighbours on $Q$. Since $a$ and $b$ have at least $n/2 + \frac{3}{2}k$ neighbours, there are at least $3k$ edges $xx^+ \in E(Q)$, where $Q = aQxx^+Qb$ and $ax^+, xb \in E(G)$. Thus there are at least $2k$ such edges that are not one of $e_1, \ldots, e_k$; let $xx^+$ be such an edge. Then $Q':=ax^+QbxQa$ is a cycle with $e(Q')>e(Q)$. Checking whether we are in this case and obtaining $Q'$ takes $O(n)$ time. 

Thus extending $P$ at most $2n$ times as described above gives us a Hamilton cycle, and the running time of the algorithm is dominated by $O(n^3)$.
\end{proof}

\begin{defn}
\label{def:sparseinstance}
For $0 < \varepsilon < 1/2$, we call a $\{0,1\}$-instance $(n,w)$ of TSP a \emph{$(n, d, \varepsilon)$-sparse instance} if $w(K_n) = d \binom{n}{2}$ with $d \leq \frac{1}{4} n^{-\frac{1}{2} - \varepsilon}$.
\end{defn}

\begin{thm}
\label{thm:sparse}
Fix $\;0 < \varepsilon < 1/2$.
There exists an $O(n^4)$-time TSP-domination algorithm which has domination ratio at least $1-6n^{-2\varepsilon}$ for any $(n,d, \varepsilon)$-sparse instance of TSP with $n \geq 225$. We call this Algorithm~C.
\end{thm}
\begin{proof}
Assume $n \geq 225$ and let $(n,w)$ be an $(n,d,\varepsilon)$-sparse instance of TSP so that $w(K_n) = d \binom{n}{2}$ with $d \leq \frac{1}{4}n^{-\frac{1}{2} - \varepsilon}$. Let $G = (V_n, E)$ be the graph where $e \in E$ if and only if $w(e)=0$. We write $\overline{G}$ for the complement of $G$.
Let
\[
S := \{ v \in V_n: \deg_G(v) \leq 2n/3 \},
\]
and let $\overline{S} := V \setminus S$.

\medskip

\noindent
\textbf{Claim.} \emph{We have $|S| \leq n^{\frac{1}{2} - \varepsilon}$. 
}

\smallskip

\noindent 
In $\overline{G}$, we have that 
$S = \{ v \in V_n: \deg_{\overline{G}}(v) \geq (n/3)-1 \}$. Then
\begin{align*}
\frac{1}{2}|S|\left( \frac{n}{3}-1 \right)  \leq e(\overline{G}) = d \binom{n}{2} \leq \frac{1}{4}n^{-\frac{1}{2} - \varepsilon}\binom{n}{2} \leq \frac{1}{8}n^{\frac{3}{2} -\varepsilon},
\end{align*}
which implies the required bounds.%
\COMMENT{Then
\[
|S| \leq \frac{\frac{1}{4}n^{\frac{3}{2} - \varepsilon}}{\frac{n}{3} - 1} \leq \frac{\frac{1}{4}n^{\frac{3}{2} - \varepsilon}}{\frac{n}{4}} = n^{\frac{1}{2} - \varepsilon},
\]
where the last inequality follows from $n \geq 12$.} 

\medskip

\noindent
We note for later that 
\begin{equation}
\label{eq:deg}
\delta( G[\overline{S}] ) \geq \frac{2}{3}n - |S| 
\geq \frac{2}{3}n - n^{\frac{1}{2} - \varepsilon} 
\geq \frac{1}{2}n + \frac{3}{2}n^{\frac{1}{2} - \varepsilon}
\geq \frac{1}{2}n + \frac{3}{2}|S|.
\end{equation}
where the penultimate inequality follows since $n \geq 225$.

We now describe the steps of our algorithm with the run time of each step in brackets; we then elaborate on each step below.

\begin{enumerate}
\item[1.]  Construct $G$. ($O(n^2)$ time)
\item[2.] Obtain the set $S \subseteq V_n$ as defined above. ($O(n^2)$ time)
\item[3.] Construct $G' := G[S , \overline{S}]$. ($O(n^2)$ time)
\item[4.] Apply the maximum double matching algorithm (Theorem~\ref{thm:mdmt}) to $G'$ to obtain a double matching $M^*$ of $G'$ from $S$ to $\overline{S}$ of maximum size. ($O(n^4)$ time)
\item[5.]  Arbitrarily extend $M^*$ to a maximum double matching $M^{**}$ of $K_n[S, \overline{S}]$ from $S$ to
 $\overline{S}$, i.e.\ $\deg_{M^{**}}(v) = 2$ for all $v \in S$,
$\deg_{M^{**}}(v) \leq 1$ for all $v \in \overline{S}$,
 and $E(M^*) \subseteq E(M^{**})$. ($O(n^2)$ time)
\item[6.] For each $v \in S$, determine its two neighbours $x_v,y_v \in \overline{S}$ in $M^{**}$, and let $e_v := x_vy_v \in \overline{S}^{(2)}$. ($O(n)$ time)
\item[7.] Apply the algorithm of Lemma~\ref{lem:dirac} to obtain a Hamilton cycle $H^*$ of $G[\overline{S}] \cup \{e_v: v \in S \}$ that includes all edges $\{e_v: v\in S\}$. ($O(n^3)$ time)
\item[8.] Replace each edge $e_v = x_vy_v$ in $H^*$ by the two edges $vx_v,vy_v \in E(M^{**})$ to obtain a Hamilton cycle $H^{**}$ of $K_n$. ($O(n)$ time)
\end{enumerate}

Note that the edges $e_v$ determined in Step~6 are independent because $M^{**}$ is a double matching. This together with (\ref{eq:deg}) means that we can indeed apply the algorithm of Lemma~\ref{lem:dirac} in Step~7. Altogether, the algorithm takes $O(n^4)$ time.

Finally, let us verify that the algorithm above gives a domination ratio of at least $1-6n^{-2 \varepsilon}$ 
 by showing that for a uniformly random Hamilton cycle $H \in \mathcal{H}_n$, we have
$\Prob(w(H) < w(H^{**})) \leq 6n^{-2 \varepsilon}$. As in the proof of Theorem~\ref{thm:dense}, it suffices to show
\begin{equation} \label{eqn:cond2}
\Prob(w(H) < w(H^{**}) \mid \mathcal{E}) = 0,
\end{equation}
where $\mathcal{E} := \mathcal{E}_1 \cap \mathcal{E}_2$ with $\mathcal{E}_1$ being the event that $H$ uses no edge of $S$ and $\mathcal{E}_2$ being the event that $e_H(v,S) \leq 1$ for all $v \in V_n \setminus S$.%
\COMMENT
{
Just as in Theorem~\ref{thm:dense}, since $|S| \leq n^{\frac{1}{2} - \varepsilon}$, we can apply Lemma~\ref{lem:S} to conclude that 
\[
\Prob(\mathcal{E}) \geq 1 - 6n^{- 2\varepsilon}.
\]
 We show that if $H$ is any Hamilton cycle in $\mathcal{E} \subseteq \mathcal{H}_n$ then $w(H^{**}) \leq w(H)$, i.e.\
\begin{equation} 
\Prob(w(H) < w(H^{**}) \mid \mathcal{E}) = 0.
\end{equation}
Assuming (\ref{eqn:cond2}), we can prove the theorem since
\begin{align*}
\Prob(w(H) < w(H^{**})) = \;
 &\Prob(w(H) < w(H^{**}) \mid \mathcal{E})\Prob(\mathcal{E}) \; + \\
&\Prob(w(H) < w(H^{**}) \mid \overline{\mathcal{E}})\Prob(\overline{\mathcal{E}}) \\
=\;&0 + \Prob(\overline{\mathcal{E}}) \leq 6n^{-2\varepsilon},
\end{align*}
as required. 
}

In order to verify (\ref{eqn:cond2}), it suffices to show that $|E(H^{**}) \cap E(G)| \geq |E(H) \cap E(G)|$ for all $H \in \mathcal{E} \subseteq \mathcal{H}_n$.
Assuming $H \in \mathcal{E}$, note that $E(H) \cap E_G(S, \overline{S})$ is a double matching of $G[S, \overline{S}]$ from $S$ to $\overline{S}$; hence by the definition of $M^*$ and $H^{**}$, we have
\[
|E(H) \cap E_G(S, \overline{S})| \leq |E(M^*)| = |E(H^{**}) \cap E_G(S, \overline{S})|.
\]
Also
\[
|E(H^{**}) \cap E_G(\overline{S})| = n - 2|S| = |E(H) \cap E_G(\overline{S})|.
\]
Using these two inequalities and the fact that $E(H) \cap E(S^{(2)}) = E(H^{**}) \cap E(S^{(2)}) = \emptyset$ (since $H, H^{**} \in \mathcal{E}$) we have
\begin{align*}
|E(H) \cap E(G)| 
&= |E(H) \cap E_G(\overline{S})| + |E(H) \cap E_G(S, \overline{S})| \\
&\leq |E(H^{**}) \cap E_G(\overline{S})| + |E(H^{**}) \cap E_G(S, \overline{S})| \\
&= |E(H^{**}) \cap E(G)|,
\end{align*}
as required.
\end{proof}

\COMMENT{Sentence has changed. Previously a theorem below changed to a lemma now with slightly changed statement}%
Finally we show that we can combine our three algorithms to give a complete algorithm for arbitrary $\{0,1\}$-instances of TSP.

\begin{lem} \label{lem:comb}
Let $\varepsilon := \frac{1}{28}$\COMMENT{Actually, I think we can do $1/14$}. There exists $n_0$ such that if $(n,w)$ is a $\{0,1\}$-instance of TSP with $n>n_0$ and $w(K_n) = d \binom{n}{2}$ for some $d \in [0,1]$, then either $(n,w)$
 is $(n,d,\varepsilon)$-regular or $(n,d,\varepsilon)$-sparse or
 $(n,d,r,\varepsilon)$-dense (for some $r$). Furthermore\COMMENT{TSP12: extra statement added}, we can decide in $O(n^4)$ time which of the three cases hold (and also determine a suitable value of $r$ in the third case).
\end{lem}
\begin{proof}
We choose $n_0 \in \mathbb{N}$ sufficiently large for our estimates to hold.\COMMENT{TSP10: previous sentence added.}
Assume that $(n,w)$ is a $\{0,1\}$-instance of TSP with $n>n_0$ which is not $(n,d,\varepsilon)$-regular. By Corollary~\ref{cor:reg}, we either have $d< 10^4(1 + 2\varepsilon) n^{-\frac{2}{3}} \log n$ or $\overline{d}:=1-d < 10^4(1 + \varepsilon)n^{-\frac{2}{7}} \log n$.

Assuming the former, we have $d \leq \frac{1}{4}n^{-\frac{1}{2} - \varepsilon}$%
\COMMENT{provided 
$10^4(1 + 2\varepsilon) n^{-\frac{2}{3}} \log n \leq \frac{1}{4}n^{-\frac{1}{2} - \varepsilon}$
or equivalently $4 \cdot 10^4(1 + 2\varepsilon) \log n \leq 
n^{\frac{11}{84}}$, which holds provided $10^5 \log n \leq 
n^{\frac{1}{8}}$. In general $k \log n \leq n^{\alpha}$ iff 
$n \leq \exp(n^{\alpha} / k)$ which holds provided $n \leq 
(n^\alpha / k)^r / r!$ i.e. $k^r r! \leq n^{\alpha r - 1}$. 
Choosing $r = \lceil \alpha^{-1} \rceil + 2$, we see that 
$k \log n \leq n^{\alpha}$ if $n > k^r r!$. In our case $\alpha = 1/8$, $k=10^5$, $r=10$ so  $n > 10! 10^{50} \leq 10^{60}$ suffices.    
}
for $n>n_0$. Thus $(n,w)$ is $(n,d, \varepsilon)$-sparse.

The only remaining possibility is that $\overline{d} < 10^4(1 + \varepsilon)n^{-\frac{2}{7}} \log n$. In addition, since $(n,w)$ is not $(n,d,\varepsilon)$-regular, we may assume that, for every optimal matching $M$ of $K_n$, we have
\[
w(M) \geq \frac{1}{2}dn - m_{\varepsilon}(n, \overline{d}),
\] 
where we recall that $m_{\varepsilon}(n, \overline{d}) = 40 (\varepsilon + \varepsilon^{\frac{1}{2}}) \log n + 40 \varepsilon^{\frac{1}{2}} \overline{d}^{\frac{1}{4}} \sqrt{n \log n}$. Setting $r := m_{\varepsilon}(n, \overline{d})$, we now verify that $(n,w)$ is $(n,d,r, \varepsilon)$-dense to complete the proof.

Parts (i) and (ii) of Definition~\ref{def:instance} clearly hold.
For part (iii) note that%
\COMMENT{TSP9: $\log^2$ replaced by $\log$ below} 
\begin{align*}
r = m_{\varepsilon}(n, \overline{d})
\leq 40\log n +  40\overline{d}^{\frac{1}{4}} \sqrt{n \log n} 
&\leq 40\log n + 
 40 \cdot 10(1 + \varepsilon)n^{\frac{3}{7}} \log n \\
&\leq 500 n^{\frac{3}{7}} \log n,
\end{align*}
and that $6\overline{d}^2(n+1) \leq 10^9(1 + \varepsilon )^2n^{\frac{3}{7}} \log^2n
\leq 2 \cdot 10^9 n^{\frac{3}{7}} \log^2n$. Using these bounds, we have
\[
6 \overline{d}^2(n+1) + 46r + 3 \overline{d} \leq 3 \cdot 10^9 n^{\frac{3}{7}} \log^2n \leq n^{\frac{1}{2} - \varepsilon},
\]
provided that $n_0$ is sufficiently large.%
\COMMENT{
$3 \cdot 10^9 n^{\frac{3}{7}} \log^2n \leq n^{\frac{1}{2} - \varepsilon}$ 
iff $3 \cdot 10^9 \log^2 n \leq n^{\frac{1}{28}}$, which holds if $10^5 \log n \leq n^{\frac{1}{56}}$. Here $k=10^5$, $\alpha = 1/56$ and $r=58$, so the inequality holds if $n > 10^{290} \cdot 58!$. I'm sure we could get something better, but is it worth it? 
}

For\COMMENT{TSP12: proof of extra statement} $\varepsilon = 1/28$ and $n \geq n_0$, one can determine in $O(n^4)$ time whether $(n,w)$ satisfies Definition~\ref{def:reginstance}, \ref{def:instance}, or \ref{def:sparseinstance}; one simply needs to determine $w(K_n)$ and an optimal matching of minimum weight using Theorem~\ref{thm:matalg}.
\end{proof}

Finally, we can prove our main result\COMMENT{TSP12: proof added}

\removelastskip\penalty55\medskip\noindent{\bf Proof of Theorem~\ref{thm:main}}.
Given an instance $(n,w)$ with $n > n_0$ and $w(K_n) = d \binom{n}{2}$ for some $d \in [0,1]$, Lemma~\ref{lem:comb} tells us that $(n,w)$ is either $(n,d,\varepsilon)$-regular or
 $(n,d,r,\varepsilon)$-dense (for some $r$) or $(n,d,\varepsilon)$-sparse, where $\varepsilon = 1/28$. Furthermore, we can determine which case holds in $O(n^4)$ time. We accordingly apply Algorithm~A, B, or C, and by Theorem~\ref{thm:ratio}, \ref{thm:dense}, or \ref{thm:sparse} respectively, we achieve a domination ratio of at least $1 - \max\{ 2n^{- \varepsilon}, 6n^{-2\varepsilon}\} \geq 1- 6n^{-1/28}$.
\endproof


\section{Hardness Results}
\label{sec:hard}

In this section\COMMENT{TSP10: First three paragraphs below changed a bit} we give upper bounds on the possible domination ratio of a polynomial-time TSP-domination algorithm for $\{0,1\}$-instances, assuming either $P \not= NP$ or the Exponential Time Hypothesis. This is achieved by a simple reduction from the Hamilton path problem to the algorithmic TSP-domination problem for $\{0,1\}$-instances. Our reduction uses an idea from \cite{GutKolYeo}, where a result similar to Theorem~\ref{thm:hardness1}(a) was proved, but for general TSP rather than $\{0,1\}$-instances of TSP.

The Hamilton path problem is the following algorithmic problem: given an instance $(n,G)$, where $n$ is a positive integer and $G$ is an $n$-vertex graph, determine whether $G$ has a Hamilton path. This problem is NP complete \cite{GarJon}, which implies Theorem~\ref{thm:NP}(a).

The Exponential Time Hypothesis \cite{IP} states that the $3$-satisfiability problem ($3$-SAT) cannot be solved in subexponential time. It would imply that there is no $\exp(o(n))$-time algorithm to solve $3$-SAT, where $n$ is the number of clauses of an input formula. Since $3$-SAT can be reduced to the Hamilton path problem in polynomial time (and space), this gives us Theorem~\ref{thm:NP}(b).
\COMMENT{TSP9: Both theorems below slightly rephrased}

\begin{thm} \label{thm:NP}
$ $
\begin{enumerate}
\item[(a)]If $P \not= NP$, then the Hamilton path problem has no polynomial-time algorithm.
\item[(b)]There is some $\varepsilon_0>0$ such that, if the Exponential Time Hypothesis holds, then there is no algorithm that solves the Hamilton path problem for all instances $(n,G)$ in time $\exp(o(n^{\varepsilon_0}))$.%
\COMMENT{Something stronger is stated formally for various NP-hard problems (e.g.\ independent set problem), but I can't find it stated anywhere for the Hamilton path problem. So I don't know the exact value $\varepsilon_0$ (I would guess > 1/100); it depends on the best reduction from the Hamilton path problem to the 3-SAT problem}
\end{enumerate}
\end{thm}

We can use Theorem~\ref{thm:NP} to obtain the hardness result Theorem~\ref{thm:hardness1} for TSP-domination ratios.\COMMENT{TSP10: sentence added}

\removelastskip\penalty55\medskip\noindent{\bf Proof of Theorem~\ref{thm:hardness1}}.
For part (a), we give a polynomial-time reduction from the Hamilton path problem. Suppose $D$ is a polynomial-time TSP-domination algorithm with domination ratio at least $1-\exp(-n^{\varepsilon})$ i.e., given a $\{0,1\}$-instance of $(n,w)$ of TSP, $D$ outputs a Hamilton cycle $H^*$ such that
\[
|\{H: w(H)<w(H^*) \}| \leq \exp(-n^{\varepsilon})|\mathcal{H}_n|.
\]
We show that we can use $D$ to determine whether any graph $G$
 has a Hamilton path in time polynomial in $|V(G)|$, contradicting Theorem~\ref{thm:NP}(a). Assume 
$G$ is an $n$-vertex graph
 and let $\overline{G}$ be its complement. Choose $n'$ such that $\lfloor n'^{\varepsilon} / \log n' \rfloor = n$; thus $n'$ is bounded above by a polynomial in $n$.  
Consider the $n'$-vertex graph $G' = (V',E')$ 
obtained from $G$ by first partitioning $V'$ into two sets 
$S, \overline{S}$ with $|S|=n$ and $|\overline{S}|= n'-n$, 
and setting $G'[S]$ to be isomorphic to $\overline{G}$, setting $G'[\overline{S}]$ to be empty, and setting $G'[S, \overline{S}]$ to be the complete bipartite 
graph between $S$ and $\overline{S}$. We define an instance $(n',w)$ of TSP where the edge weighting $w$ of $K_{n'} \supseteq G'=(V',E')$ is given by $w(e)=1$ for each $e \in E'$ and $w(e)=0$ for $e \not\in E'$.

Applying $D$ to $(n',w)$, we obtain a Hamilton cycle $H^*$ of $K_{n'}$ in time polynomial in $n'$ (which is therefore polynomial in $n$). We claim that $G$ has a Hamilton path if and only if $w(H^*)=2$, and this proves Theorem~\ref{thm:hardness1}(a).

To see the claim, first assume that $w(H^*)=2$.
Note that for any $H \in \mathcal{H}_{n'}$, we have $w(H[S, 
\overline{S}]) \geq 2$, and if we have equality, then $H[S]$ 
is a Hamilton path of $K_{n'}[S]$ (and $H[\overline{S}]$ is a 
Hamilton path of $K_{n'}[\overline{S}]$). Since $w(H^*)=2$, we 
have $w(H[S, \overline{S}]) = 2$ and $w(H^*[S])=0$. Thus 
$H^*[S]$ is a Hamilton path of $K_{n'}[S]$ and every edge of 
$H^*[S]$ is an edge of $G$. Hence $G$ has a Hamilton 
path.

To see the claim in the other direction, suppose $G$ has a Hamilton path $P$ and consider the set $\mathcal{H}_P$ of all Hamilton cycles $H$ of $K_{n'}$ such that $H[S] = P$. Note that $w(H)=2$ for all $H \in \mathcal{H}_P$ and that 
\begin{align} \label{eq:hard}
|\mathcal{H}_P| 
= (n'-n)! &= 2[(n'-1) \cdots (n'-n+1)]^{-1}|\mathcal{H}_{n'}|   \\
& \geq 2 (n')^{-n}|\mathcal{H}_{n'}| 
= 2 \exp(- n \log n')|\mathcal{H}_{n'}|  \nonumber \\
&> \exp( -n'^{\varepsilon} )|\mathcal{H}_{n'}|; \nonumber
\end{align}
the last inequality follows by our choice of $n'$. Since $D$ has domination ratio at least $1-\exp(-n'^{\varepsilon})$ for $(n',w)$, $D$ must output a Hamilton cycle $H^*$ with $w(H^*)=2$. This proves the claim and completes the proof of (a).

The proof of (b) follows almost exactly the same argument as above. Set $C := 2 \varepsilon_0^{-1} + 2$, where $\varepsilon_0$ is as in the statement of Theorem~\ref{thm:NP}(b), and set $\beta(n) := \exp(-(\log n)^C)$.

Suppose $D$ is an $O(n^k)$-time algorithm (for some $k \in \mathbb{N}$) which, given a $\{0,1\}$-instance of $(n,w)$ of TSP, outputs a Hamilton cycle $H^*$ such that
\[
|\{H: w(H)<w(H^*) \}| \leq \beta(n)|\mathcal{H}_n|.
\]
We show that we can use $D$ to determine whether any $n$-vertex graph $G$
 has a Hamilton path in time $\exp(o(n^{\varepsilon_0}))$, contradicting Theorem~\ref{thm:NP}(b). We set $n' := \lceil \exp(n^{\varepsilon_0/2}/k) \rceil$ and apply $D$ to the $n'$-vertex graph $G'$ constructed in the same way as in part (a). The algorithm takes time $O(n'^k)=\exp(o(n^{\varepsilon_0}))$ and 
the proof of the claim proceeds in the same say as in part (a), except that (\ref{eq:hard}) is replaced by $|\mathcal{H}_P| > \beta(n')|\mathcal{H}_{n'}|$.%
\COMMENT{
Set $n' := \lceil \exp(n^{ \varepsilon_0 / 2}/k) \rceil$ so that $n \leq (k \log n')^{2/\varepsilon_0 + 1}$ if $n$ is sufficiently large.
Applying $D$ to $(n',w)$, we obtain a Hamilton cycle $H^*$ of $K_{n'}$ in time $O(n'^k) = \exp(o(n^{\varepsilon_0}))$. As before we claim $G$ has a Hamilton path if and only if $w(H^*)=2$ (and $n_0$ large enough as a function of $k$; if not then we can find a Hamilton path of $G$ in constant time), and this proves the theorem.
To see the claim, the forward implication is exactly as before
To see the claim in the other direction, suppose $G$ has a Hamilton path $P$ and consider the set $\mathcal{H}_P$ of all Hamilton cycles $H$ of $K_{n'}$ such that $H[S] = P$. Note that $w(H)=2$ for all $H \in \mathcal{H}_P$ and that 
\begin{align*}
|\mathcal{H}_P| 
= (n'-n)! &= 2[(n'-1) \cdots (n'-n+1)]^{-1}|\mathcal{H}_{n'}|   \\
& \geq 2 (n')^{-n}|\mathcal{H}_{n'}| 
= 2 \exp(- n \log n')|\mathcal{H}_{n'}| \\
&> \exp( - (\log n')^{2\varepsilon_0^{-1}+2} )|\mathcal{H}_{n'}|;
\end{align*}
the final inequality follows by our choice of $n'$. Since $D$ has domination ratio at least $1-\exp( - (\log n')^{2\varepsilon_0^{-1}+2} )$ for $(n',w)$, $D$ must output a Hamilton cycle $H^*$ with $w(H^*)=2$. 
} 
\endproof






\section{Concluding remarks and an open problem}
\label{sec:conc}

In view of Theorem~\ref{thm:main}, it is natural to ask whether there exists a poly\-nomial-time TSP-domination algorithm with domination ratio $1-o(1)$ for general instances of TSP.
Below we discuss some of the challenges which this would involve. 

Note that the average weight of a perfect matching in a weighted complete graph $K_n$ is $w(K_n)/(n-1)$.
For instances $(n,w)$ of TSP for which there exists a perfect matching $M^*$ of $K_n$ with $w(M^*)$ `significantly' smaller than $w(K_n)/(n-1)$, our techniques immediately imply that Algorithm~A gives a large domination ratio. We call these \emph{typical instances}.  However instances of TSP where all perfect matchings have approximately the same weight (after some suitable normalisation) need to be treated separately, and we call these \emph{special instances}. For weights in $\{0,1\}$, there are essentially two different types of special instances: either most edges have weight $0$ or most edges have weight $1$. Even then, some additional ideas were needed to treat the special instances. If we allow
weights in $[-1,1]$, there are many different types of special instances and the challenge is to find a way to treat all of these cases simultaneously. As an example of a special instance, consider a weighting of $K_n$ where we take $S \subseteq V_n$ with $|S|$ close to $n/2$ and set $w(e):=1$ if $e \in S^{(2)}$, $w(e):=-1$ if $e \in \overline{S}^{(2)}$, and $w(e):=0$ if $e \in E(S, \overline{S})$. Then any perfect matching has weight close to zero.\COMMENT{We actually know what all special instances: is it worth mentioning this?}

\medskip

{\footnotesize \obeylines \parindent=0pt

Daniela K\"{u}hn, Deryk Osthus, Viresh Patel 
School of Mathematics 
University of Birmingham
Edgbaston
Birmingham
B15 2TT
UK
\begin{flushleft}
{\it{E-mail addresses}:\\
\tt{\{d.kuhn, d.osthus\}@bham.ac.uk}, viresh.s.patel@gmail.com}
\end{flushleft}}

\end{document}